\documentclass[journal,10pt]{IEEEtran}

\usepackage[noadjust]{cite}

\ifCLASSINFOpdf
\usepackage[pdftex]{graphicx}
\graphicspath{{../pdf/}{../jpeg/}}
\DeclareGraphicsExtensions{.pdf,.jpeg,.png}
\else
\usepackage[dvips]{graphicx}
\graphicspath{{../eps/}}
\DeclareGraphicsExtensions{.eps}
\fi

\usepackage[labelformat=simple, font=footnotesize]{subcaption}

\usepackage{amsmath}
\usepackage{amssymb}
\usepackage{bm}

\usepackage{epstopdf}
\newcommand\numberthis{\addtocounter{equation}{1}\tag{\theequation}}
\usepackage[normalem]{ulem}
\usepackage[binary-units = true]{siunitx}

\usepackage{tikz}
\usepackage{pgfplots}
\usetikzlibrary{shapes, shapes.geometric, arrows}

\usepackage{amsthm}
\newtheoremstyle{exampstyle}
{3pt} 
{0pt} 
{} 
{} 
{\bfseries} 
{.} 
{.5em} 
{} 
\theoremstyle{exampstyle} 
\newtheorem{theorem}{Theorem}
\theoremstyle{exampstyle} 
\newtheorem{lemma}{Lemma}
\theoremstyle{exampstyle} 

\theoremstyle{exampstyle} 
\newtheorem{problem}{Problem}
\theoremstyle{exampstyle}

\usepackage[ruled,linesnumbered]{algorithm2e}

\usepackage[font=footnotesize,labelsep=period]{caption}

\DeclareMathOperator*{\argmax}{arg\,max}

\usepackage{enumitem}

\begin{document}

\title{Multihop Routing for Data Delivery in V2X Networks}

\author{Yilin~Li,~\IEEEmembership{Student~Member,~IEEE,}
	Jian~Luo,
	Richard A. Stirling-Gallacher,~\IEEEmembership{Member,~IEEE,}
	Zhongfeng~Li,
	and~Giuseppe~Caire,~\IEEEmembership{Fellow,~IEEE}
	\thanks{Yilin Li is with the German Research Center, Huawei Technologies Duesseldorf GmbH, 80992 Munich, Germany, and with the Communications and Information Theory Group, Technische Universit{\"a}t Berlin, 10587 Berlin, Germany (e-mail: halodiplomat@gmail.com).}%
	\thanks{Jian Luo, Richard A. Stirling-Gallacher, and Zhongfeng Li are with the German Research Center, Huawei Technologies Duesseldorf GmbH, 80992 Munich, Germany (e-mail: jianluo@huawei.com; richard.sg@huawei.com; lizhongfeng@huawei.com).}%
	\thanks{Giuseppe Caire is with the Communications and Information Theory Group, Technische Universit{\"a}t Berlin, 10587 Berlin, Germany, and with the Department of Electrical Engineering, The University of Southern California, Los Angeles, CA 90089, USA (e-mail: caire@tu-berlin.de).}%
}

\maketitle

\begin{abstract}
Data delivery relying on the carry-and-forward strategy of vehicle-to-vehicle (V2V) communications is of significant importance, however highly challenging due to frequent connection disruption. Fortunately, incorporating vehicle-to-infrastructure (V2I) communications, motivated by its availability in bridging long-range vehicular connectivity, dramatically improves delivery opportunity. Nevertheless, the cooperation of V2V and V2I communications, known as vehicular-to-everything (V2X) communications, necessitates a specific design of multihop routing for enhancing data delivery performance. To address this issue, this paper provides a mathematical framework to investigate the data delivery performance in V2X networks in terms of both delivery latency and data rate. With theoretical analysis, we formulate a global and a distributed optimization problem to maximize the weighted sum of delivery latency and data rate. The optimization problems are then solved by convex optimization theory and based on the solutions, we propose a global and a distributed multihop routing algorithm to select the optimal route for maximizing the weighted sum. The rigorousness of the proposed algorithms is validated by extensive simulation under a wide range of system parameters and simulation results shed insight on the design of multihop routing algorithm in V2X networks for minimizing latency and maximizing data rate.
\end{abstract}

%
%
\IEEEpeerreviewmaketitle

\section{Introduction}\label{Ch4_Sec1}
Intelligent transportation system (ITS) revolutionizes the provisioning of diverse applications associated with driving safety, traffic management, and infotainment~\cite{Karagiannis}. These applications are beyond the far-fetched goals of academic, industry, and standardization groups, which aim to streamline the innovative operation of vehicle, facilitate safe and eco-friendly driving, and offer ubiquitous infotainment services for commuting passengers~\cite{Zeadally}. 

Typically, potential connectivity disruption as a result of high vehicle speed, time-varying vehicle density, and limited inter-vehicle contact time, confines vehicle-to-vehicle (V2V) communications to applications and services with short communication range~\cite{Kenney}. Fortunately, vehicle-to-infrastructure (V2I) communications represent a viable solution to bridge long-range vehicular connectivity by introducing stationary network entities, e.\,g., road side unit (RSU), to exchange data with vehicles~\cite{3GPPV2X}. By leveraging the complementary features of infrastructure-less V2V communications and infrastructure-assisted V2I communications, the use of hybrid vehicular communication, namely vehicle-to-everything (V2X), has been envisioned as a full-fledged solution to capture the connectivity, efficiency, and scalability of vehicular networks~\cite{Abdrabou}. 

Fifth generation (5G) mobile communications target to support efficient data delivery with bulky data rate, high reliability, and low latency for different V2X services and applications~\cite{YLiCommag}. Specifically, real-time applications, such as collision avoidance and lane-change/merge notification, bring the requirement of low latency. In addition, non-emergency services, such as video streaming, demand high data rate to fulfill capacity burst. Therefore, an efficient data delivery scheme to meet the diverse requirements of V2X communications is expected to be designed properly taking into account the coexistence of low latency and high data rate~\cite{Reis}.

Nevertheless, data delivery in V2X networks is particularly challenging, as its performance depends highly on the efficiency of data routing~\cite{JLiu}. On the one hand, pure inter-vehicle data delivery may introduce non-negligible latency because of frequent disconnection. On the other hand, the limited coverage of each RSU is a major concern of pure inter-RSU data delivery. Therefore, the design of multihop routing algorithm that minimizes latency as well as maximizes data rate becomes an interesting and challenging topic.

\subsection{Related Works}\label{Ch4_Sec1_RW}
Multihop routing for data delivery, particularly for vehicular networks, has been intensively investigated by recent research efforts~\cite{FZhang, YZhu, LYao, Choi, GSun, Ni, Alsharif, YWu, MWang1, JChen, JChen1, MWang, Atallah, JHeTMC1, MXingTMC, MXingTVT, JHeTMC2}. The knowledge of vehicular trajectories plays a key role for optimal data delivery, where the performance of routing algorithm relies heavily on the accuracy of vehicle mobility prediction~\cite{FZhang, YZhu, LYao, Choi, GSun}. 

It has been envisioned that the assistance of infrastructure facilitates data delivery by latency improvement~\cite{Ni, Alsharif, YWu, MWang1}. However, these works either assumed that the latency of V2V and V2I transmission can be ignored, or limited the latency to be considered in a single hop between vehicles, which is not proper for data delivery in large-scale vehicular networks where multihop transmissions are expected. More importantly, models applied to the above works depend on the prerequisite that the size of packets transmitted on V2V and V2I link is small enough, such that the data rate for delivering these packets is omitted. This assumption does not hold for the data delivery of services that rely on abundant data rate to guarantee the enormous requirement of data volume.

When it comes to achieving reasonable data rate for data delivery in V2X networks, a rich body of earlier studies have tackled the problem of how ``mobility improves data rate'' in vehicular networks~\cite{JChen, JChen1, MWang, Atallah}. With the exception of some studies that contributed to a limited investigation of latency performance~\cite{JChen}, none of the above-mentioned works to date has considered the trade-off between low latency and high data rate of data delivery in V2X networks, which is supposed to be a key enabler in fully exploit the mobility of vehicles complemented by the stability of infrastructures to improve data rate performance while keep latency tolerable.

Store-carry-and-forward strategy, where data is stored at intermediate nodes along delivery and forwarded at a later time to another intermediate station or the final destination, has been recently recognized as a promising evolution path to improve data delivery efficiency, either in latency or data rate~\cite{JHeTMC1, MXingTMC, MXingTVT, JHeTMC2}. Nevertheless, except the dropbox functionality that allows data to be stored with some cost, none of the aforementioned studies has considered a realistic model of RSU for providing V2I communications in terms of link establishment and resource allocation. Moreover, the ability of V2X networks to support connections for other devices besides vehicular users (e.g.\ cellular users), which is one of the key considerations in optimizing the overall system performance in this paper, has not been addressed in any of these works.  

\subsection{Contributions}\label{Ch4_Sec2_Ctrbt}
Different from the previous studies~\cite{JHeTMC1, MXingTMC}, where RSU is simply assumed to be a dropbox for data collection and temporary storage, in this paper we consider RSU as a network entity that supports both V2I and cellular communications and serves both vehicular and cellular users. 
The main contributions of this paper are summarized as follows:
\begin{itemize}
	\item Development of an accurate analytical framework for data delivery in V2X networks: The framework first considers \textit{hop-wise} latency and data rate, which are derived taking into account divergent vehicle mobility patterns and data forwarding behaviors at each hop. Based on these, the expected \textit{end-to-end} (E2E) latency and data rate are obtained by adding the hop-wise latency and minimizing the hop-wise data rate, respectively. 
	\item Formulation of optimization problems that maximize the weighted sum of latency and data rate: Unlike some of the previous works that revolved around the feasibility study of data delivery without explicitly addressing performance optimization (\cite{YZhu, Alsharif, MWang1}), we obtain mathematical expressions of both hop-wise and E2E latency/data rate, based on rigorous derivations, and formulate optimization problems that maximize the weighted sum of latency and data rate considering both global and distributed scenarios, where the weighted sum is optimized in E2E manner and hop-wise manner, respectively. 
	\item Leveraging the optimization problems and the corresponding solutions to propose multihop routing algorithms: The derived expressions of latency and data rate are transformed into closed-form for verifying the convexity of the proposed optimization problems, which are then solved by convex optimization theory. Based on these, multihop routing algorithms to select the optimal route are proposed for both global and distributed data delivery in terms of the weighted sum maximization.
	\item A detailed system-level performance evaluation for data delivery: 
	Extensive simulations are conducted under numerous system parameters to demonstrate the efficiency of the proposed algorithms in achieving lower latency and higher data rate compared to classical vehicular routing algorithms. The impact of broadcast scheme, vehicle arrival rate, and backhaul availability on the delivery performance are also analyzed.
\end{itemize}

The remainder of this paper is organized as follows: Section~\ref{Ch4_Sec2} presents the system model and Section~\ref{Ch4_Sec3} addresses the optimization problem formulation. In Section~\ref{Ch4_Sec4}, we solve the formulated problems and propose the corresponding routing algorithms. The proposed algorithms are then evaluated by extensive simulations in Section~\ref{Ch4_Sec5}, followed by a summary concluding this paper in Section~\ref{Ch4_Sec6}.

A conference version of this paper has appeared in~\cite{YLiGC}. The current paper extends the previous work with the development of distributed data delivery optimization, the design of distributed routing algorithm, and the support of wireless backhaul in data delivery. The current paper also includes all derivations, discussions of extensions, and more detailed simulations.

\section{System Model}\label{Ch4_Sec2}
In this section, we introduce the network, traffic, data forwarding, and radio models considered for finding the best route and optimizing the data delivery. 

\subsection{Network Model}\label{Ch4_Sec2_SubSec1}
We envision a scalable V2X network for data delivery incorporating both RSU-assisted and carry-and-forward strategies. The network is geographically and equally divided by the coverage of RSUs. Data is generated by a source node $S$ and to be carried and forwarded by vehicles to a destination node $D$. In general, $S$ and $D$ can be either RSU or vehicle. Without loss of generality, we assume that both $S$ and $D$ are RSUs. The traffic information of all vehicles in the networks are available at RSUs, however not all RSUs are necessarily interconnected\footnote{The data delivery problem with fully interconnected RSUs can be solved by routing algorithms for fixed network topology, which have been thoroughly studied and are trivial for V2X networks. Nevertheless, in this paper we also consider partially and/or fully interconnected RSUs that enable data forwarding via wireless backhaul between RSUs in Section~\ref{Ch4_Sec5}.}. We further assume $n$ routes, denoted as a set $\bm{\Gamma} = \{\gamma_i|i=1,\dotsc,n\}$, between $S$ and $D$. A route $\gamma_i$ could be composed of multiple hops and each hop refers to a road segment within the coverage of corresponding RSU.

\subsection{Traffic Model}\label{Ch4_Sec2_SubSec2}
Similar to the previous works~\cite{JChen, Reis, JHeTMC1, MXingTMC}, we assume that vehicles move at the same speed and stay within each hop for a constant duration $T$\footnote{This assumption holds well for highway or rural area, where RSUs are equidistantly deployed and vehicles on each lane move at the same speed with slight deviation. For urban scenario where diverse coverages and speed limits are expected, the networks can be partitioned into blocks and in each block, vehicles are likely to move at the similar speed due to speed limit along the isometric road segment between RSUs, hence the assumption is also valid.}. We further adopt the widely used traffic model where the
number of vehicles arrives at a hop and heads to the next hop is Poisson distributed~\cite{JChen, Reis}.

\subsection{Data Forwarding Model}\label{Ch4_Sec2_SubSec3}
Consider one route out of the route set $\bm{\Gamma}$. As there might not exist a single vehicle that moves along all hops of the considered route, the data would need to be forwarded by the vehicle that carries the data, when it no longer heads to $D$, to another vehicle that does. We referred to the vehicle that carries data and moves along a hop of the route as the \textit{courier} of the hop. If the courier moves towards the next hop of the route when leaving the current hop (the information about whether moving towards the next hop or not can be acquired by e.g.\ navigation or pre-configured route of autonomous vehicle), the data is carried by the courier to the next hop and consequently the courier of the next hop is still this vehicle. Otherwise, when arriving at the current hop, the courier tries to discover another vehicle that moves towards the next hop, referred to as the \textit{candidate} of the current hop, to forward the data. The discovery maintains at most a duration $t$ (the time for discovery process cannot exceed $t$), referred to as the \textit{global candidate discovery duration}, and once succeeds, the courier forwards the data to the candidate via V2V communications within the remaining time it stays in the hop. Otherwise, it sends the data to the RSU that covers the hop via V2I communications within $T-t$, then the RSU will find a suitable candidate heading to the next hop and forward the data. 

\subsection{Radio Model}\label{Ch4_Sec2_SubSec4}
For candidate discovery, we further assume that when arriving at a hop, the courier starts to broadcast beacon for the candidate discovery assisted by RSU. Information encapsulated in the beacon could include, e.g., the direction of the next hop. Once the beacon has been received and successfully decoded, with error probability $\epsilon$, the candidate sends feedback to report the successful reception. The feedback is then decoded by the courier, again with error probability $\epsilon$, and in case of incorrect decoding or no feedback detected, the courier repeats the discovery trial until a communication link between the courier and the candidate is established within $t$. 

\section{Problem Formulation}\label{Ch4_Sec3}
Without loss of generality, in the following sections we focus on a single route, which is referred to as the \textit{typical} route and randomly picked from the route set $\bm{\Gamma}$, between $S$ and $D$. We further denote the number of hops of the typical route as $k$ and the $h$-th hop of the typical route as hop $h$, respectively. Analysis of other routes can be derived similarly. 

\subsection{End-to-End Latency}\label{Ch4_Sec3_SubSec1}
As mentioned in Section~\ref{Ch4_Sec2_SubSec2}, the arrival of vehicles at a hop follows a Poisson distribution. Here, we denote the arrival rate of vehicles that arrive at hop $h$ and head to hop $h+1$ as $\lambda_{h,h+1}$. Based on the traffic model and data forwarding model described in Section~\ref{Ch4_Sec2_SubSec2} and Section~\ref{Ch4_Sec2_SubSec3}, respectively, the E2E latency of the typical route is described in the following scenarios. Examples including courier moves to the next stop, successful discovery, and failed discovery are illustrated in hop 1, 2, and 3 of Fig.~\ref{Ch4_Fig_disc}, respectively.
\begin{figure}[tbp]
	\centering%
	\includegraphics [width=\columnwidth]{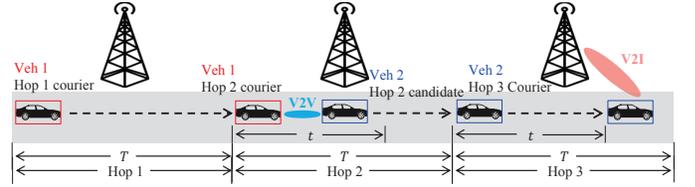}%
	\caption{Illustration of a route with 3 hops that correspond to courier moves to the next hop, successful discovery, and failed discovery, respectively.}\label{Ch4_Fig_disc}%
\end{figure}

\subsubsection{Courier Moves Towards the Next Hop}\label{Ch4_Sec3_SubSec1_SubSubSec1}~\\
In this scenario, the courier carries the data to the next hop (hop $h+1$), as depicted in hop 1 of Fig.~\ref{Ch4_Fig_disc}. Hence, there is no need for candidate discovery. Denoting the number of way-out directions except U-turn for hop $h$ as $\text{Deg}_h$. Then, the probability of the event ``Courier of hop $h$ moves towards hop $h+1$'', denoted as $P(\text{Courier:}\, h \rightarrow h+1)$, satisfies
\setlength{\abovedisplayskip}{3pt}
\setlength{\belowdisplayskip}{3pt}
\begin{equation} 
P(\text{Courier:}\, h \rightarrow h+1) = \frac{1}{\text{Deg}_h},
\label{eqnPcourier}
\end{equation}
which means when leaving hop $h$, the courier randomly selects a direction with equal probability. Clearly, the hop-wise latency of hop $h$ for the event ``Courier of hop $h$ moves towards hop $h+1$'', denoted as $L_{ \text{Courier:}\, h \rightarrow h+1 }$, is the duration of the courier staying in the hop, namely 
\begin{equation} 
L_{ \text{Courier:}\, h \rightarrow h+1 } = T.
\label{eqnDcourier}
\end{equation}

\subsubsection{Courier Succeeds in Candidate Discovery}\label{Ch4_Sec3_SubSec1_SubSubSec2}~\\
In case the courier is not heading to the next hop, a candidate discovery for data forwarding is carried out by the courier within the discovery duration $t$, as illustrated in hop 2 of Fig.~\ref{Ch4_Fig_disc}. When the number of arrivals in a given time interval follows Poisson distribution, inter-arrival times are known to have the exponential distribution~\cite{Kleinrock}. Thus, the probability of the event ``A candidate heading to hop $h+1$ arrives at hop $h$ within $t$ on condition that courier of hop $h$ does not move to hop $h+1$'', denoted as $P(\tau_{h,h+1} \leq t)$, satisfies
\begin{align*} 
P(\tau_{h,h+1} \leq t) 
&=  \int_0^t \! \lambda_{h,h+1}e^{-\lambda_{h,h+1}\tau_{h,h+1}} \, \mathrm{d}\tau_{h,h+1} \\
&\qquad \cdot \big( 1-P(\text{Courier:}\, h \rightarrow h+1) \big) \\
&= \big( 1-e^{-\lambda_{h,h+1}t} \big) \bigg( 1-\frac{1}{\text{Deg}_h} \bigg), \numberthis
\label{eqnPCandidate}
\end{align*}
where $\tau_{h,h+1}$ represents the time between the courier arrives at hop $h$ and a candidate heading to $h+1$ arrives at hop $h$.

When arriving at hop $h$, the courier starts to broadcast beacon to discover a candidate for data forwarding, as described in Section~\ref{Ch4_Sec2_SubSec4}. We denote the time of one discovery trial as $\Delta t$, which includes beacon broadcasting and feedback receiving time. Then, within the candidate discovery duration $t$, there would be maximally $m=\bigl\lfloor\frac{t}{\Delta t}\bigr\rfloor$ rounds of discovery trials, where $\lfloor . \rfloor$ represents the floor function. 
Then, the probability of the event ``Courier of hop $h$ successfully discovers a candidate'', denoted as $P(\text{Success})$, satisfies
\begin{align*} 
P(\text{Success}) &= P(\tau_{h,h+1} \leq t) \cdot \sum_{i=0}^{m-1} \big(1-(1-\epsilon)^2\big)^{i}{(1-\epsilon)^2} \\
&= \bigg( 1-\frac{1}{\text{Deg}_h} \bigg) \big( 1-e^{-\lambda_{h,h+1}t} \big) \\
&\qquad \cdot \Big(1-\big(1-(1-\epsilon)^2\big)^m\Big). \numberthis
\label{eqnPfind}
\end{align*}

Similarly to the previous scenario, the hop-wise latency of hop $h$ for the event ``Courier of hop $h$ successfully discovers a candidate'', denoted as $L_{\text{Success}}$, satisfies
\begin{equation} 
L_{\text{Success}} = T.
\label{eqnDfind}
\end{equation}

\subsubsection{Courier Fails in Candidate Discovery}\label{Ch4_Sec3_SubSec1_SubSubSec3}~\\
In this scenario, the courier has to send the data to RSU as no proper candidate can be discovered within $t$, as drawn in hop 3 of Fig.~\ref{Ch4_Fig_disc}. Accordingly, the probability of the event ``No candidate towards hop $h+1$ arrives at hop $h$ within $t$ on condition that courier of hop $h$ does not move to hop $h+1$'', denoted as $P(\tau_{h,h+1} > t)$, satisfies
\begin{align*} 
P(\tau_{h,h+1} > t) 
&= e^{-\lambda_{h,h+1}t} \bigg( 1-\frac{1}{\text{Deg}_h} \bigg). \numberthis
\label{eqnPnCandidate}
\end{align*}
Then, the probability of the event ``Courier of hop $h$ fails to discover a candidate'', denoted as $P(\text{Failure})$, can be derived as
\begin{align*} 
P(\text{Failure}) &= P(\tau_{h,h+1} \leq t) \big(1-(1-\epsilon)^2\big)^m + P(\tau_{h,h+1} > t) \\
&= \bigg(1-\frac{1}{\text{Deg}_h}\bigg) \big(1-e^{-\lambda_{h,h+1}t}\big) \big(1-(1-\epsilon)^2\big)^m\\
& \qquad + \bigg( 1-\frac{1}{\text{Deg}_h} \bigg) e^{-\lambda_{h,h+1}t}. \numberthis
\label{eqnPnfind}
\end{align*}

As the data is not forwarded to any candidate within $t$, the courier will transmit the data to RSU in the remaining time $T-t$. Specifically, the courier has to stay in the hop for $T$ anyway to move through the hop, which means that except the time for candidate discovery $t$, the remaining time for the courier to transmit the data to RSU is $T-t$. Afterwards, the RSU will find an appropriate candidate to forward the data and once found, the candidate receives the data from the RSU when moving along the hop within $T$. Therefore, the summarized hop-wise latency of hop $h$ for the event ``Courier of hop $h$ fails to discover a candidate'', denoted as $L_{\text{Failure}}$, is calculated as 
\begin{equation} 
L_{\text{Failure}} = T + \tau^{(d,\,\text{RSU})}_{h,h+1} + T = 2T + \tau^{(d,\,\text{RSU})}_{h,h+1}.
\label{eqnDnfind}
\end{equation}
Here, $\tau^{(d,\,\text{RSU})}_{h,h+1}$ represents the time for RSU to find a candidate. Note that the first $T$ does not overlap with $\tau^{(d,\,\text{RSU})}_{h,h+1}$, as in the remaining time after the courier has failed in candidate discovery, i.e.\ $T-t$, the RSU is receiving the data from the courier and not able to start to find a candidate.

Combining these scenarios, the expected hop-wise latency for data delivery, denoted as $E(L_{h,h+1})$, can be derived as
\begin{align*}
E(L_{h,h+1}) &= P(\text{Courier:}\, h \rightarrow h+1) E(L_{ \text{Courier:}\, h \rightarrow h+1 })\\
&\quad + P(\text{Success}) E(L_{\text{Success}}) \\
&\quad + P(\text{Failure}) E(L_{\text{Failure}}). \numberthis
\label{eqnDhop}
\end{align*}
In particular, $\tau^{(d,\,\text{RSU})}_{h,h+1}$ is also exponentially distributed due to Poisson distribution of vehicle arrivals, and correspondingly we have
\begin{align*}
E(L_{\text{failure}}) 
& = 2T+\frac{1}{\lambda_{h,h+1}}. \numberthis
\label{eqnDnfindmean}
\end{align*}

Finally, the expected E2E latency of the typical route, denoted as $\bar{L}$, is provided as
\begin{equation} 
\bar{L} = \sum_{h=1}^{k} E(L_{h,h+1}).
\label{eqnDroute}
\end{equation}

\subsection{End-to-End Data Rate}\label{Ch4_Sec3_SubSec2}
As mentioned in Section~\ref{Ch4_Sec1}, we consider a realistic model of RSU in terms of providing communication links to both vehicular and cellular users, which is different from the dropbox functionality considered in~\cite{JHeTMC1, MXingTMC, MXingTVT, JHeTMC2} that only allows data to be stored with some cost.
Similar to Section~\ref{Ch4_Sec3_SubSec1}, three scenarios are addressed here for the analysis of the expected data rate of the typical route.

\subsubsection{Courier Moves Towards the Next Hop}\label{Ch4_Sec3_SubSec2_SubSubSec1}~\\
In this scenario, RSU is not requested by courier for assisting candidate discovery. Therefore, the RSU exclusively serves cellular users. Denoting the achievable data rate of the services provided for cellular users as $r_\text{O}$, the achievable hop-wise data rate of hop $h$ for the event ``Courier of hop $h$ moves towards hop $h+1$'', denoted as $C_{ \text{Courier:}\, h \rightarrow h+1 }$, is calculated as
\begin{equation}
C_{ \text{Courier:}\, h \rightarrow h+1 } = r_\text{O}. 
\label{eqnUcourier}
\end{equation}
Note that $r_\text{O}$ indicates the overall achievable data rate and can be obtained by e.g.\ taking average of data rates among all cellular users in the network.


\subsubsection{Courier Succeeds in Candidate Discovery}\label{Ch4_Sec3_SubSec2_SubSubSec2}~\\
When a communication link has been established between courier and candidate, the carried data is transmitted from the courier to the candidate via V2V communications with data rate $r_\text{V2V}$, which can be similarly obtained as $r_\text{O}$ by e.g.\ taking average of data rates among all  V2V communication links. Denoting the actual candidate discovery time for the courier of hop $h$ as $\tau_{h,h+1}^{(d,\,\text{Veh})}$, the achievable hop-wise data rate of hop $h$ for the event ``Courier of hop $h$ successfully discovers a candidate'', denoted as $C_{\text{Success}}$, can be written as
\begin{equation} 
C_{\text{Success}} = \frac{r_\text{V2V}\Big(T-\tau_{h,h+1}^{(d,\,\text{Veh})}\Big)}{T} + \frac{r_\text{O}(T-t)}{T}. \label{eqnUfind}
\end{equation}
Here, the data rate consists of two parts. The first term of the right hand side of~\eqref{eqnUfind} indicates that the V2V communication link maintains a duration of $T-\tau_{h,h+1}^{(d,\,\text{Veh})}$, and the amount of data can be transmitted is calculated as $r_\text{V2V}(T-\tau_{h,h+1}^{(d,\,\text{Veh})})$. As the RSU of hop $h$ is not aware of the actual candidate discovery time $\tau_{h,h+1}^{(d,\,\text{Veh})}$, it preserves $t$ for the candidate discovery and correspondingly, the data rate provided by the RSU for cellular users is calculated as $r_\text{O}(T-t)$.

\subsubsection{Courier Fails in Candidate Discovery}\label{Ch4_Sec3_SubSec2_SubSubSec3}~\\
In this scenario, the data is first transmitted to RSU after failed candidate discovery within $t$ and then forwarded to a proper candidate. The RSU of hop $h$ is correspondingly exclusively associated with the courier for candidate discovery and data reception, until the data is successfully forwarded to a candidate. Therefore, the amount of data transmitted from the courier to the RSU via V2I communications is $r_\text{V2I}(T-t)$, where $r_\text{V2I}$ indicates the data rate of V2I communications and can be obtained similarly to $r_\text{V2V}$ and $r_\text{O}$. After finding a candidate within $\tau^{(d,\,\text{RSU})}_{h,h+1}$ and forwarding the received data to the candidate within $T-t$ (the data rate between the candidate and the RSU is also assumed to be $r_\text{V2I}$), the RSU is able to serve cellular users with amount of data $r_\text{O}t$. In summary, the achievable hop-wise data rate of hop $h$ for the event ``Courier of hop $h$ fails to discover a candidate'', denoted as $C_{\text{Failure}}$, can be written as
\begin{equation} 
C_{\text{Failure}} = \frac{ r_\text{V2I}(T-t) }{2T+\tau^{(d,\,\text{RSU})}_{h,h+1}} + \frac{ r_\text{O}t }{2T+\tau^{(d,\,\text{RSU})}_{h,h+1}}. \label{eqnUnfind}
\end{equation}

Combining these scenarios, the expected hop-wise data rate for data delivery, denoted as $E(C_{h,h+1})$, can be derived as
\begin{align}
E(C_{h,h+1}) &= P(\text{Courier:}\, h \rightarrow h+1) E(C_{ \text{Courier:}\, h \rightarrow h+1 })\notag \\
&\quad + P(\text{Success}) E(C_{\text{Success}})\notag \\
&\quad + P(\text{Failure}) E(C_{\text{Failure}}).
\label{eqnUhop}
\end{align}

For a multihop route between $S$ and $D$, the achievable data rate is determined by the ``weakest'' hop in which the lowest rate is achieved. Therefore, the expected E2E data rate of the typical route, denoted as $\bar{C}$, is calculated as
\begin{equation} 
\bar{C} = \min_{\forall h} E(C_{h,h+1}).
\label{eqnUroute}
\end{equation}

\subsection{Global Data Delivery Problem}\label{Ch4_Sec3_SubSec3}
Based on the analysis in Section~\ref{Ch4_Sec3_SubSec1} and Section~\ref{Ch4_Sec3_SubSec2}, it is clear that the global candidate discovery duration $t$ plays a vital role in determining the E2E latency and data rate. On the one hand, a larger $t$ allows courier a better chance to find a candidate, however leaves less time for data forwarding, which leads to data rate degrade. On the other hand, a decreased $t$ brings more time for data forwarding and correspondingly enhances the achievable data rate, while an increased latency is expected due to less opportunity for successful candidate discovery. Therefore, a trade-off between latency and data rate to optimize the overall performance can be achieved by the adaptation of $t$. The global data delivery optimization problem is formulated as follows:
\begin{problem}
	(Global weighted sum maximization)
	\begin{equation}
	\max_{\substack{t \in [0,T]}} \alpha \bar{C} - (1-\alpha) \bar{L}.
	\label{eqnOpt}
	\end{equation}\label{Ch4_prb1}
\end{problem}
Here, $\alpha \in [0,1]$ is the weight parameter. Distinguished by various use cases and application services, $\alpha$ can be flexibly adjusted, e.\,g., $\alpha = 1$ may refer to latency-tolerant but rate-sensitive use case, while $\alpha = 0$ may indicate real-time services which are keen to latency with relative low demand on data rate. Here, the term ``global'' indicates the universal configuration of the candidate discovery duration, where an identical candidate discovery duration $t$ is applied to all hops in the typical route. Note that here $\bar{C}$ and $\bar{L}$ are normalized in $[0,1]$, where the motivation of the normalization lies in the fact that data ranges and units of latency and data rate are not directly comparable~\cite{Marler}.

\subsection{Distributed Data Delivery Problem}\label{Ch4_Sec3_SubSec4}
In addition to the global data delivery problem, we further propose the distributed data delivery problem, where the weighted sum of latency and data rate is hop-wisely maximized, compared to the E2E-wise maximization addressed in Section~\ref{Ch4_Sec3_SubSec3}. The motivation of the distributed data delivery comes from the fact that the arrival rates $\lambda_{h,h+1}$, the actual candidate discovery time for RSU $\tau_{h,h+1}^{(d, \text{RSU})}$, and the actual candidate discovery time for vehicle $\tau_{h,h+1}^{(d, \text{Veh})}$, are diverse in each hop. Therefore, hop-tailored maximization may benefit from a hop-specific configuration of discovery duration, as the hop-individual weighted sum could bring potential gain in increasing data rate while reducing latency.

By replacing $t$ in~\eqref{eqnPCandidate},~\eqref{eqnPfind},~\eqref{eqnPnfind},~\eqref{eqnUfind}, and~\eqref{eqnUnfind} with $\hat{t}_h$, which we refer to as the \textit{hop-wise candidate discovery duration} of hop $h$, the expected hop-wise latency and data rate, denoted as $\hat{L}_{h}$ and $\hat{C}_{h}$, respectively, satisfy
\begin{align*}
\hat{L}_{h} &= P(\text{Courier:}\, h \rightarrow h+1) T + \hat{P}(\text{Success}) T \\
&\quad + \hat{P}(\text{Failure}) \bigg(2T+\frac{1}{\lambda_{h,h+1}}\bigg), \numberthis
\label{eqnDhopDstr}
\end{align*}
and
\begin{align*}
\hat{C}_{h} &= P(\text{Courier:}\, h \rightarrow h+1)r_\text{O} + \hat{P}(\text{Success}) E(\hat{C}_{\text{Success},h}) \\
&\quad + \hat{P}(\text{Failure}) E(\hat{C}_{\text{Failure},h}). \numberthis
\label{eqnUhopDstr}
\end{align*}
where the terms with hat in~\eqref{eqnDhopDstr} and~\eqref{eqnUhopDstr} indicates corresponding expressions in~\eqref{eqnPCandidate}--\eqref{eqnUnfind} by replacing $t$ with $\hat{t}_h$.

Finally, the distributed data delivery optimization problem is formulated as follows:
\begin{problem}
	(Distributed weighted sum maximization)
	\begin{equation}
	\max_{\substack{\hat{t}_h  \in [0,T]}} \alpha \hat{C}_{h} - (1-\alpha) \hat{L}_{h}.
	\label{eqnOptDstr}
	\end{equation}\label{Ch4_prb2}
\end{problem}
Here, $\alpha \in [0,1]$ is the weight parameter. Similar to the global data delivery problem, a trade-off between latency and data rate to optimize the overall performance can be achieved by the adaptation of $\hat{t}_h$.

\section{Data Delivery Optimization and Routing Algorithm Design}\label{Ch4_Sec4}
In this section, we propose solutions of the data delivery optimization problems formulated in Section~\ref{Ch4_Sec3}. 

\subsection{Reformation of Problem Formulation}\label{Ch4_Sec4_SubSec1}
The optimization problem~\ref{Ch4_prb1} and~\ref{Ch4_prb2} are formulated in a sophisticated way in Section~\ref{Ch4_Sec3}, and it is relatively hard to verify the convexity of the optimization problems determined in~\eqref{eqnOpt} and~\eqref{eqnOptDstr}. In this subsection, the derived latency and data rate are reformed into closed-form expressions.

\subsubsection{Reformation of the End-to-End Latency}\label{Ch4_Sec4_SubSec1_SubSubSec1}~\\
For simplicity, let $\alpha_h =1/\text{Deg}_h$, $\beta_h(t) = e^{-\lambda_{h,h+1}t}$, $\theta_h(t) = \big(1-(1-\epsilon)^2\big)^m$, and $\phi_h = T+1/\lambda_{h,h+1}$. Then, we transform the expected hop-wise latency derived in~\eqref{eqnDhop} as
\begin{align*}
E(L_{h,h+1}) 
&= T + (1-\alpha_h)\phi_h\big(\beta_h(t)+\theta_h(t)-\beta_h(t)\theta_h(t)\big) . \numberthis
\label{eqnDhopNew}
\end{align*}
Based on this, the expected E2E latency, which is depicted in~\eqref{eqnDroute}, is now calculated as
\begin{align*}
\bar{L} 
&= kT+ \sum_{h=1}^{k} (1-\alpha_h)\phi_h\big(\beta_h(t)+\theta_h(t) -\beta_h(t)\theta_h(t)\big). \numberthis
\label{eqnDrouteNew}
\end{align*}

\subsubsection{Reformation of the End-to-End Data Rate}\label{Ch4_Sec4_SubSec1_SubSubSec2}~\\
For simplicity, let $\zeta_h = (1-\alpha_h)r_\text{O}$, $\iota_h = \frac{ r_\text{V2V}(T-1/\lambda_{h,h+1}) }{T}+r_\text{O}$, $\kappa_h = \frac{r_\text{V2I}}{2T+1/\lambda_{h,h+1}}T$, $\nu_h(t) = -\frac{r_\text{O}}{T}t$, $\chi_h(t) = \frac{r_\text{O}-r_\text{V2I}}{2T+1/\lambda_{h,h+1}}t$, and $z(t)=\beta_h(t)+\theta_h(t)-\beta_h(t)\theta_h(t)$. Then, similar to the reformation of the expected hop-wise latency, the expected hop-wise data rate, which is derived in~\eqref{eqnUhop}, can be transformed as
\begin{align*}
E(C_{h,h+1}) 
&= \zeta_h + (1-\alpha_h)\big(1-z(t)\big)  \big(\iota_h + \nu_h(t) \big)\\
& \quad + (1-\alpha_h)z(t) \big( \kappa_h + \chi_h(t) \big). \numberthis
\label{eqnUhopNew}
\end{align*}

As the expression of $E(C_{h,h+1})$ in~\eqref{eqnUhopNew} is a combination of multiplication and summation of hop-dependent terms (terms with the subscript $h$), finding the closed-form of the minimum of $E(C_{h,h+1})$, namely $\bar{C}$, is still a very complicated problem. Therefore, we further reformed $\bar{C}$ as
\begin{align*}
\bar{C} &= P(\text{All success})E(C_{\text{All success}})\\
& \quad + P(\text{All failure})E(C_{\text{All failure}})\\
& \quad + P(\text{Mixture})E(C_{\text{Mixture}}), \numberthis
\label{eqnUroute1}
\end{align*}
where
\begin{align*}
P(\text{All success}) &= \textstyle\prod_{h=1}^{k} P(\text{Success}), \numberthis
\label{eqnPallfind}
\end{align*}
\begin{align*}
P(\text{All failure}) &= \textstyle\prod_{h=1}^{k} P(\text{Failure}), \numberthis
\label{eqnPnonefind}
\end{align*}
and
\begin{align*}
P(\text{Mixture}) &= 1 - P(\text{All success}) - P(\text{All failure}). \numberthis \label{eqnPsomefind}
\end{align*}
Similarly, $C(\text{All success})$, $C(\text{All failure})$, and $C(\text{Mixture})$ represent the expected data rate of the corresponding events and are provided by 
\begin{align*}
C_{\text{All success}} &= \min_{\forall h} C_{\text{Success},h}, \numberthis
\label{eqnUallfind}
\end{align*}
\begin{align*}
C_{\text{All failure}} &= \min_{\forall h} C_{\text{Failure},h}, \numberthis
\label{eqnUnonefind}
\end{align*}
and
\begin{align*}
C_{\text{Mixture}} &= \min_{\substack{\forall h,l, h \neq l}} (r_\text{O}, C_{\text{Success},h}, C_{\text{Failure},l}). \numberthis \label{eqnUsomefind}
\end{align*}
Now, it is clear that the main task of finding the closed-form of the expected E2E data rate is to reform the minimization operators addressed in $C(\text{All success})$, $C(\text{All failure})$, and $C(\text{Mixture})$.

a) Closed-form of $E(C_{\text{All success}})$.
\begin{lemma}
	Given a set of i.i.d geometrically distributed random variables $\{\textrm{X}_i | i=1,\dotsc,n\}$ with parameter $p \in (0,1)$, the probability mass function (PMF) of $\textrm{Y}=\max(\textrm{X}_1,\ldots,\textrm{X}_n)$, denoted as $f_{\textrm{Y}}{(x)}$, satisfies
	\begin{equation} 
	f_{\textrm{Y}}{(x)} = \big(1 - (1-p)^x \big) ^n - \big(1 - (1-p)^{x-1} \big) ^n. \label{eqnPDFmaxG}
	\end{equation}\label{Ch4_Lem1}
\end{lemma}
\begin{proof}
	The CDF of the maximum $\textit{\textrm{Y}}$ can be calculated as: 
	\begin{align*}
	P(\textit{\textrm{Y}} \leq x) &= \displaystyle\prod_{i=1}^{n} P(\textit{\textrm{X}}_i \leq x) = P(\textit{\textrm{X}} \leq x)^n \\
	&= \big(1 - (1-p)^x \big) ^n. \numberthis
	\label{eqnCDFmaxG}
	\end{align*}
	Hence, the PMF of $\textit{\textrm{Y}}$ can be derived as
	\begin{align*}
	f_{\textrm{Y}}{(x)} &= P(\textit{\textrm{Y}} \leq x) - P\big(\textit{\textrm{Y}} \leq (x-1)\big)\\
	&= \big(1 - (1-p)^x \big) ^n - \big(1 - (1-p)^{x-1} \big) ^n . \numberthis
	\label{eqnPDFmaxGderi}
	\end{align*}
\end{proof}
Considering the scenario of successful candidate discovery in all hops, the E2E data rate is limited by the ``weakest'' hop, at which the longest time for candidate discovery is consumed. Consequently, the goal of finding $\min_{\forall h} C_{\text{Success},h}$ turns to find $\max_{\forall h} \tau_{h,h+1}^{(d,\,\text{Veh})}$, which is solved by the following theorem.
\begin{theorem}
	Let $m_h\in \{1, \dotsc, m\}, \forall h \in \{1,\dotsc,k\}$ and $\xi = \max_{\forall h} m_h$ denote the actual number of discovery trials in hop $h$ and the maximum of $m_h$, respectively. Let $p=(1-\epsilon)^2$. Then, the closed-form of $E(C_{\text{All success}})$ satisfies
	\begin{equation} 
	E(C_{\text{All success}}) = \frac{r_\textnormal{V2V} \bigg(T - E\Big(\max_{\forall h} \tau_{h,h+1}^{(d,\,\textnormal{Veh})} \Big) \bigg) + r_\text{O}(T-t) }{T}, \label{eqnUallfindtheo}
	\end{equation}
	where
	\begin{align*}
	E\Big(\max_{\forall h} \tau_{h,h+1}^{(d,\,\textnormal{Veh})} \Big) &= \sum_{\xi=1}^{m} \xi \Big( \big(1 - (1-p)^{\xi} \big) ^k\\
	&\qquad \quad - \big(1 - (1-p)^{{\xi}-1} \big) ^k \Big) \cdot \Delta t. \numberthis
	\label{eqnUallfindmtheo}
	\end{align*}\label{Ch4_Thm1}
\end{theorem}
\begin{proof}
	$E(C_{\text{All success}})$ can be derived as
	\begin{align*}
	E(C_{\text{All success}}) 
	&= \frac{r_\text{V2V}\bigg(T - E\Big(\max_{\forall h} \tau_{h,h+1}^{(d,\,\text{Veh})} \Big) \bigg) + r_\text{O}(T-t) }{T}, \numberthis
	\label{eqnUallfindderi1}
	\end{align*}
	where $\tau_{h,h+1}^{(d,\,\text{Veh})} = m_h \Delta t$. As $\Delta t$ remains constant, finding $\max_{\forall h} \tau_{h,h+1}^{(d)}$ turns further to find $\xi = \max_{\forall h} m_h$.
	Given Lemma~\ref{Ch4_Lem1}, the PMF of $\xi$, denoted as $f{(\xi)}$, is calculated as
	\begin{equation}
	f{(\xi)} = \big(1 - (1-p)^{\xi} \big) ^k - \big(1 - (1-p)^{{\xi}-1} \big) ^k. 
	\label{eqnUallfindderi2}
	\end{equation}
	Then $\forall h \in \{1,\dotsc,k\}$, we have
	\begin{align*}
	E\Big(\max_{\forall h} \tau_{h,h+1}^{(d,\,\text{Veh})} \Big) &=  \sum_{\xi=1}^{m} \xi \Delta t \cdot f{(\xi)}\\
	 &= \sum_{\xi=1}^{m} \xi \Big( \big(1 - (1-p)^{\xi} \big) ^k\\
	 &\qquad \quad - \big(1 - (1-p)^{{\xi}-1} \big) ^k \Big) \cdot \Delta t. \numberthis
	\label{eqnUallfindderi3}
	\end{align*}
\end{proof}

b) Closed-form of $E(C_{\text{All failure}})$.~\\
The E2E data rate of this scenario can be solved similarly. Specifically, finding $\min_{\forall h} C_{\text{Failure},h}$ turns to find $\max_{\forall h} \tau^{(d,\,\text{RSU})}_{h,h+1}$ where $\tau^{(d,\,\text{RSU})}_{h,h+1}$ follows a exponentially distribution with parameter $\lambda_{h,h+1}$. 
\begin{lemma}
	Given a set of i.i.d exponentially distributed random variables $\{\textrm{X}_i | i=1,\dotsc,n\}$ with parameter $\lambda_i$, the probability density function (PDF) of $\textrm{Z}=\max(\textrm{X}_1,\dotsc,\textrm{X}_n)$, denoted as $f_{\textrm{Z}}{(x)}$, satisfies
	\begin{equation} 
	f_{\textrm{Z}}{(x)} = \sum_{i=1}^{n} \Bigg( \lambda_i e^{-\lambda_i x} \displaystyle\prod_{\substack{j=1, j \neq i}}^{n} \big( 1-e^{-\lambda_j x} \big) \Bigg). \label{eqnPDFmaxE}
	\end{equation}\label{Ch4_Lem2}
\end{lemma}
\begin{proof}
	The CDF of the maximum $\textit{\textrm{Z}}$ can be calculated as: 
	\begin{align*}
	p(\textit{\textrm{Z}} \leq x) &= \displaystyle\prod_{i=1}^{n} p(\textit{\textrm{X}}_i \leq x) = \displaystyle\prod_{i=1}^{n} \big(1-e^{-\lambda_i x} \big). \numberthis
	\label{eqnCDFmaxE}
	\end{align*}
	Hence, the PDF of $\textit{\textrm{Z}}$ can be derived as
	\begin{align*}
	f_{\textrm{Z}}{(x)} &= \frac{\mathrm{d}}{\mathrm{d}x} p(\textit{\textrm{Z}} \leq x)
	= \sum_{i=1}^{n} \Bigg( \lambda_i e^{-\lambda_i x} \displaystyle\prod_{\substack{j=1, j \neq i}}^{n} \big( 1-e^{-\lambda_j x} \big) \Bigg). \numberthis
	\label{eqnPDFmaxEderi}
	\end{align*}
\end{proof}
Given Lemma~\ref{Ch4_Lem2}, the closed-form of the expected E2E data rate $E(C_{\text{All failure}})$ is solved by the following theorem.
\begin{theorem}
	Let $\eta = \max_{\forall h} \tau^{(d,\,\textnormal{RSU})}_{h,h+1} \in [0,\infty)$ and $\mu_h=\lambda_{h,h+1}, \forall h \in \{1,\dotsc,k\}$, respectively. Then, the closed-form of $E(C_{\text{All failure}})$ satisfies
	\begin{align*}
	E(C_{\text{All failure}}) &= \frac{r_\textnormal{V2I}(T - t)+r_\text{O}t}{2T+E\Big(\max_{\forall h} \tau^{(d,\,\textnormal{RSU})}_{h,h+1}\Big)}, \numberthis
	\label{eqnUnonefindtheo}
	\end{align*}
	where
	\begin{align*}
	E(\max_{\forall h}\tau^{(d,\,\textnormal{RSU})}_{h,h+1}) &= \sum_{h=1}^{k} \mu_h \int_0^\infty \! \eta \Bigg( e^{-\mu_h \eta} \\
	&\qquad \qquad \qquad \cdot \displaystyle\prod_{\substack{l=1, l \neq h}}^{k} \big( 1-e^{-\mu_l \eta} \big) \Bigg) \mathrm{d}\eta. \numberthis
	\label{eqnUnonefindmtheo}
	\end{align*}\label{Ch4_Thm2}
\end{theorem}
\begin{proof}
	$E(C_{\text{All faliure}})$ can be derived as
	\begin{align*}
	E(C_{\text{All failure}}) 
	&= \frac{r_\text{V2I}(T - t)+r_\text{O}t}{2T+E\Big(\max_{\forall h}\tau^{(d,\,\text{RSU})}_{h,h+1}\Big)}. \numberthis
	\label{eqnUallonederi1}
	\end{align*}
	Given Lemma~\ref{Ch4_Lem2}, the PDF of $\eta$, denoted as $f(\eta), \eta \in [0,\infty), \forall h \in \{1,\dotsc,k\}$, can be written as
	\begin{equation}
	f{(\eta)} = \sum_{h=1}^{k} \Bigg( \mu_h e^{-\mu_h \eta} \displaystyle\prod_{\substack{l=1, l \neq h}}^{k} \big( 1-e^{-\mu_l \eta} \big) \Bigg). 
	\label{eqnUnonefindderi2}
	\end{equation}
	Then, we have
	\begin{align*}
	E\Big(\max_{\forall h} \tau^{(d,\,\text{RSU})}_{h,h+1} \Big) &= E(\eta ) = \int_0^\infty \! \eta f{(\eta)} \mathrm{d}\eta \\
	&= \sum_{h=1}^{k} \mu_h \int_0^\infty \! \eta \Bigg( e^{-\mu_h \eta} \\
	&\qquad \qquad \qquad \cdot \displaystyle\prod_{\substack{l=1, l \neq h}}^{k} \big( 1-e^{-\mu_l \eta} \big) \Bigg)\mathrm{d}\eta. \numberthis
	\label{eqnUnonefindderi3}
	\end{align*}
\end{proof}
c) Closed-form of $E(C_{\text{Mixture}})$.~\\
The derivation of the closed-form of the expected E2E data rate $E( C_{\text{Mixture}} )$ is a bit tricky as hop latency for this scenario follows a combination of geometric distribution (successful discovery) and exponential distribution (failed discovery). To solve the problem, we first introduce the following lemma:
\begin{lemma}
	For a random variable $\textrm{X}$ with non-negative values, the expectation of $\textrm{X}$, denoted as $E(\textrm{X})$, satisfies
	\begin{equation} 
	E(\textrm{X}) = \int_0^\infty \! \big( 1-F_\textrm{X}(x) \big) \mathrm{d}x, \label{eqnExp}
	\end{equation}
	where $F_\textrm{X}(x)$ indicates the CDF of $\textrm{X}$.\label{Ch4_Lem3}
\end{lemma}
\begin{proof}
	Let $f_\textrm{X}(y)$ represents the PDF of $\textrm{X}$. Then,
	\begin{align*}
	\int_0^\infty \! \big( 1 - F_\textrm{X}(x) \big) \mathrm{d}x &= \int_0^\infty \! P(\textrm{X} \geq x) \mathrm{d}x\\
	&= \int_0^\infty \! \int_x^\infty \! f_\textrm{X}(y) \mathrm{d}y \mathrm{d}x\\
	&= \int_0^\infty \! x f_\textrm{X}(x) \mathrm{d}x = E(\textrm{X}). \numberthis
	\label{eqnExpderi2}
	\end{align*}
\end{proof}
Further, we notice that $E( C_{\text{Mixture}} )$ can be derived as
\begin{align*}
E( C_{\text{Mixture}} ) &= E\Big( \min_{\substack{\forall h,l, h \neq l}} (r_\text{O}, C_{\text{Success},h}, C_{\text{Failure},l}) \Big). \numberthis
\label{eqnUsomefindderi1}
\end{align*}
Let $\rho=\min_{\forall h,l, h \neq l} (C_{\text{Success},h}, C_{\text{Failure},l})$, then we have
\begin{align*}
E( C_{\text{Mixture}} ) &= P(r_\text{O}\leq \rho)r_\text{O} + P(r_\text{O}\geq \rho)E(\rho) \\
&= \big(1-P(\rho\leq r_\text{O})\big)r_\text{O} + P(\rho\leq r_\text{O})E(\rho) \\
&= \big( 1 - F_{\min_{\forall h} C_{\text{Success},h}}(r_\text{O}) F_{\min_{\forall l} C_{\text{Failure},l}}(r_\text{O}) \big) r_\text{O} \\
&\quad + F_{\min_{\forall h} C_{\text{Success},h}}(r_\text{O}) F_{\min_{\forall l} C_{\text{Failure},l}}(r_\text{O})E(\rho). \numberthis
\label{eqnUsomefindderi2}
\end{align*}
Here, $E(\rho)$ can be derived by applying Lemma~\ref{Ch4_Lem3} as
\begin{align*}
E(\rho) &= \int_0^\infty \! \big( 1-F_{\rho}(x) \big) \mathrm{d}x, \numberthis
\label{eqnUsomefindderi3}
\end{align*}
where
\begin{align*}
F_{\rho}(x) &= 1-\big(1-F_{\min_{\forall h} C_{\text{Success},h}}(x)\big) \big(1-F_{\min_{\forall l} C_{\text{Failure},l}}(x)\big) \\
&= F_{\min_{\forall h} C_{\text{Success},h})}(x) + F_{\min_{\forall l} C_{\text{Failure},l}}(x) \\
&\quad -F_{\min_{\forall h} C_{\text{Success},h}}(x) F_{\min_{\forall l} C_{\text{Failure},l}}(x). \numberthis
\label{eqnUsomefindderi4}
\end{align*}
In~\eqref{eqnUsomefindderi2} and~\eqref{eqnUsomefindderi4}, $F_{\min_{\forall h} C_{\text{Success},h}}(x)$ and $F_{\min_{\substack{\forall l}} C_{\text{Failure},l}}(x)$ can be calculated from $F_{\xi}(x)$ and $F_{\eta}(x)$, which can be obtained by taking summation of~\eqref{eqnUallfindderi2} and integral of~\eqref{eqnUnonefindderi2}, respectively. 

\subsection{Optimal Data Delivery}\label{Ch4_Sec4_SubSec2}
\subsubsection{Global Data Delivery Optimization}\label{Ch4_Sec4_SubSec2_SubSubSec1}~\\
From~\eqref{eqnDrouteNew},~\eqref{eqnUallfindtheo},~\eqref{eqnUallfindmtheo},~\eqref{eqnUnonefindtheo},~\eqref{eqnUnonefindmtheo},~\eqref{eqnUsomefindderi2},~\eqref{eqnUsomefindderi3}, and~\eqref{eqnUsomefindderi4}, it is evident that the global candidate discovery duration $t$ is the variable for adapting the E2E latency and data rate, given the typical route with $k$ hops, hop duration $T$, decode error rate $\epsilon$, vehicle arrival rate $\lambda_{h,h+1}$, and the data rates $r_\text{V2V}$, $r_\text{V2I}$, and $r_\text{O}$. It is shown in Appendix~\ref{Ch4_Sec7_SubSec6} that the global optimization problem is convex and differentiable, which can be solved by convex optimization theory.
\begin{theorem}
	Let $t_s$ be the stationary point of $\alpha \bar{C} - (1-\alpha) \bar{L}$, i.\,e., $\frac{\mathrm{d}(\alpha \bar{C} - (1-\alpha) \bar{L})}{\mathrm{d}t_s}=0$. Then, $t^*$ is the optimal solution for Problem~\ref{Ch4_prb1}, where
	\begin{equation}
	t^* = \argmax_{t=\{0,t_s,T\}} \alpha \bar{C} - (1-\alpha) \bar{L}, t_s \in [0,T].
	\label{eqnSolc1}
	\end{equation}\label{Ch4_Thm3}
\end{theorem}
\begin{proof}
	See Appendix~\ref{Ch4_Sec7_SubSec7}.
\end{proof}

\subsubsection{Distributed Data Delivery Optimization}\label{Ch4_Sec4_SubSec2_SubSubSec2}~\\
Similar to the global data delivery problem, it is evident that the hop-wise candidate discovery duration $\hat{t}_h$ controls the hop-wise latency and data rate considering hop duration $T$, decode error rate $\epsilon$, vehicle arrival rate $\lambda_{h,h+1}$, and the data rates $r_\text{V2V}$, $r_\text{V2I}$, and $r_\text{O}$. Specifically, the expected hop latency and data rate, which are $\hat{L}_{h}$ and $\hat{C}_{h}$ addressed in~\eqref{eqnDhopDstr} and~\eqref{eqnUhopDstr}, respectively, can be derived as
\begin{align*}
\hat{L}_{h} &= T + (1-\alpha_h)\phi_h\big(\hat{\beta_h}(t)+\hat{\theta}(t)-\hat{\beta_h}(t)\hat{\theta}(t)\big), \numberthis
\label{eqnDhopDstrSol}
\end{align*}
and
\begin{align*}
\hat{C}_{h} 
&= \zeta_h + (1-\alpha_h)\big(1-\hat{z}(t)\big) \big(\iota_h + \hat{\nu_h}(t) \big)\\
& \quad + (1-\alpha_h)\hat{z}(t) \big( \kappa_h + \hat{\chi_h}(t) \big). \numberthis
\label{eqnUhopDstrSol}
\end{align*}
where $\hat{\beta_h(t)} = e^{-\lambda_{h,h+1}\hat{t}_h}$, $\theta_h(t) = \big(1-(1-\epsilon)^2\big)^{\bigl\lfloor\frac{\hat{t}_h}{\Delta t}\bigr\rfloor}$, $\hat{\nu_h(t)} = -\frac{r_\text{O}}{T}\hat{t}_h$, $\hat{\chi_h(t)} = \frac{r_\text{O}-r_\text{V2I}}{2T+\frac{1}{\lambda_{h,h+1}}}\hat{t}_h$, and $\hat{z}(t)=\hat{\beta_h}(t)+\hat{\theta}(t)-\hat{\beta_h}(t)\hat{\theta}(t)$.

According to~\cite{YLiTWC}, $\Delta t$, which is the time of one discovery trial, depends on beam duration, frame length, and error probability, which are independent of the duration $\hat{t}_h$. Therefore, the distributed data delivery performance in terms of the weighted sum of the hop-wise latency and the hop-wise data rate, is maximized by determining the optimal $\hat{t}_h$. The convexity of the distributed optimization problem can be verified similarly to the global optimization problem~\ref{Ch4_prb1} and is omitted here to avoid redundancy.

Ultimately, the optimization problem~\ref{Ch4_prb2} is solved by the following theorem:
\begin{theorem}
	Let $\hat{t}_{h,s}$ be the stationary point of $\alpha \hat{C}_{h} - (1-\alpha) \hat{L}_{h}$, i.\,e., $\frac{\mathrm{d}(\alpha \hat{C}_{h} - (1-\alpha) \hat{L}_{h})}{\mathrm{d}\hat{t}_{h,s}}=0$. Then $\hat{t}^*_h$ is the optimal solution for Problem~\ref{Ch4_prb2}, where
	\begin{equation}
	\hat{t}^*_h = \argmax_{\hat{t}_h=\{0,\hat{t}_{h,s},T\}} \alpha \hat{C}_{h} - (1-\alpha) \hat{L}_{h}, \hat{t}_{h,s} \in [0,T].
	\label{eqnSold1}
	\end{equation}\label{Ch4_Thm4}
\end{theorem}

\subsection{Routing Algorithm Design}\label{Ch4_Sec4_SubSec3}
We summarize our routing algorithms to solve both the global data delivery optimization problem~\ref{Ch4_prb1} and the distributed data delivery optimization problem~\ref{Ch4_prb2}, based on Theorem~\ref{Ch4_Thm3} and Theorem~\ref{Ch4_Thm4}, in Algorithm~\ref{Ch4_Alg1} and Algorithm~\ref{Ch4_Alg2}, respectively. The algorithms can be e.g.\ executed at RSUs where vehicles have access to the routing information when entering in the corresponding coverage. In Algorithm~\ref{Ch4_Alg1}, the set of routes $\bm{\Gamma}$ are planned and created considering all possible routes between $S$ and $D$. Routes are iteratively selected from $\bm{\Gamma}$ for calculating the weighted sum of latency and data rate until all routes have been traversed, as indicated in line~\ref{Ch4_Alg1:1}. For each route, the weighted sum $\alpha \bar{C}_i - (1-\alpha) \bar{L}_i$ and the corresponding optimal candidate discovery duration $t_i$ are obtained in line~\ref{Ch4_Alg1:2} and \ref{Ch4_Alg1:3}, respectively. The update of the overall maximal weighted sum $opt$, the optimal route $\gamma^*$, and the overall global optimal duration $t^*$ are described in line~\ref{Ch4_Alg1:4}--\ref{Ch4_Alg1:5}. Similar procedures can be found in Algorithm~\ref{Ch4_Alg2}.

\begin{algorithm}[tbp]
	\renewcommand\baselinestretch{0.7}\selectfont
	\SetAlgoLined
	\KwIn{The set of routes $\bm{\Gamma} = \{\gamma_i|i=1,\dotsc,n\}$}
	\KwOut{The maximal weighted sum $opt$, the optimal route $\gamma^*$, and the global optimal candidate discovery duration $t^*$}
	\begin{itemize}
		{\item $n$: Number of routes between $S$ and $D$}
		{\item $\lambda_{h, \, {h+1}}^{(i)}$: Arrival rate at $h$-th hop of route $\gamma_i$}
		{\item $t_i$: Global candidate discovery duration of route $\gamma_i$}
		{\item $\bar{L}_i$: Expected E2E latency of route $\gamma_i$}
		{\item $\bar{C}_i$: Expected E2E data rate of route $\gamma_i$}
		{\item $i$: Iterator}
	\end{itemize}
	\textit{Initialization}:
	$opt=0$, $\gamma^*=\gamma_1$, $t^*=0$, $i=1$ \\
	\Begin{
		\everypar={\nl}
		\While{$i \neq n$}{\label{Ch4_Alg1:1}
			Find $\lambda_{h, \, {h+1}}^{(i)}$, $\forall h \in \{1,\dotsc,k_i\}$ in $\gamma_i$\;
			Obtain $\bar{L}_i$ and $\bar{C}_i$ by computing~\eqref{eqnDrouteNew},~\eqref{eqnUallfindtheo},~\eqref{eqnUallfindmtheo},~\eqref{eqnUnonefindtheo},~\eqref{eqnUnonefindmtheo},~\eqref{eqnUsomefindderi2},~\eqref{eqnUsomefindderi3}, and~\eqref{eqnUsomefindderi4}\; \label{Ch4_Alg1:2}
			Obtain $\alpha \bar{C}_i - (1-\alpha) \bar{L}_i$ and $t_i$ by solving $\frac{\mathrm{d}(\alpha \bar{C}_i - (1-\alpha) \bar{L}_i)}{\mathrm{d}t_i}=0$\; \label{Ch4_Alg1:3}
			\If{$\alpha \bar{C}_i - (1-\alpha) \bar{L}_i \geq opt$}{\label{Ch4_Alg1:4}
				$opt = \alpha \bar{C}_i - (1-\alpha) \bar{L}_i$\;
				$\gamma^* = \gamma_i$\;
				$t^* = t_i$\;
			}\label{Ch4_Alg1:5}
			$i=i+1$\;
		}
		\textbf{Return} $opt$, $\gamma^*$ and $t^*$\;
	}
	\caption{Global Optimization Algorithm}\label{Ch4_Alg1}
\end{algorithm}
\begin{algorithm}[tbp]
	\renewcommand\baselinestretch{0.7}\selectfont
	\SetAlgoLined
	\KwIn{The set of routes $\bm{\Gamma} = \{\gamma_i|i=1,\dotsc,n\}$}
	\KwOut{The maximal weighted sum $opt$, the optimal route $\gamma^*$, and the set of hop-wise optimal candidate discovery durations $\hat{\bm{t}}^*$}
	\begin{itemize}
		{\item $n$: Number of routes between $S$ and $D$}
		{\item $\lambda_{h, \, {h+1}}^{(i)}$: Arrival rate at $h$-th hop of route $\gamma_i$}
		{\item $\hat{t}_{h}$: Hop-wise candidate discovery duration in hop $h$}
		{\item $\hat{L}^{(i)}_{h}$: Expected hop-wise latency in hop $h$ of route $\gamma_i$}
		{\item $\hat{C}^{(i)}_{h}$: Expected hop-wise data rate in hop $h$ of route $\gamma_i$}
		{\item $\bar{L}_i$: Expected E2E latency of route $\gamma_i$}
		{\item $\bar{C}_i$: Expected E2E data rate of route $\gamma_i$}
		{\item $\hat{\textit{\textbf{t}}}_i$: Set of hop-wise candidate discovery durations of route $\gamma_i$}
		{\item $i$: Iterator}
	\end{itemize}
	\textit{Initialization}:
	$opt=0$, $\gamma^*=\gamma_1$, $\hat{\bm{t}}^*=\{\hat{t}_{h}=0|h=1,\dotsc,k_i\}$, $i=1$ \\
	\Begin{
		\everypar={\nl}
		\While{$i \neq n$}{\label{Alg2:1}
			Find $\lambda_{h, \, {h+1}}^{(i)}$, $\forall h \in \{1,\dotsc,k_i\}$ in $\gamma_i$\;
			Obtain $\hat{L}^{(i)}_{h}$ and $\hat{C}^{(i)}_{h}$ by computing~\eqref{eqnDhopDstrSol} and~\eqref{eqnUhopDstrSol}, $\forall h \in \{1,\dotsc,k_i\}$\; \label{Ch4_Alg2:3}
			Obtain $\alpha \hat{C}^{(i)}_{h} - (1-\alpha) \hat{L}^{(i)}_{h}$ and $\hat{t}_{h}$ by solving $\frac{\mathrm{d}\big(\hat{C}^{(i)}_{h} - (1-\alpha) \hat{L}^{(i)}_{h}\big)}{\mathrm{d}\hat{t}_{h}}=0$, $\forall h \in \{1,\dotsc,k_i\}$\; \label{Ch4_Alg2:4}
			Obtain $\bar{L}_i$, $\bar{C}_i$, and $\hat{\textit{\textbf{t}}}_i$ by $\bar{L}_i = \sum_{h=1}^{k_i} \hat{L}^{(i)}_{h}$, $\bar{C}_i = \min_{\forall h \in \{1,\dotsc,k_i\}} \hat{C}^{(i)}_{h}$, and $\hat{\textit{\textbf{t}}}_i = \{\hat{t}_{h}|h=1,\dotsc,k_i\}$\; \label{Ch4_Alg2:5}
			\If{$\alpha \bar{C}_i - (1-\alpha) \bar{L}_i \geq opt$}{\label{Ch4_Alg2:6}
				$opt = \alpha \bar{C}_i - (1-\alpha) \bar{L}_i$\;
				$\gamma^* = \gamma_i$\;
				$\hat{\textit{\textbf{t}}}^* = \hat{\textit{\textbf{t}}}_i$\;
			}\label{Ch4_Alg2:7}
			$i=i+1$\;
		}
		\textbf{Return} $opt$, $\gamma^*$, and $\hat{\bm{t}}^*$\;
	}
	\caption{Distributed Optimization Algorithm}\label{Ch4_Alg2}
\end{algorithm}

\section{Numerical Evaluation and Discussion}\label{Ch4_Sec5}
In this section, we evaluate the performance of the proposed routing algorithms for data delivery in V2X networks. 

\subsection{Simulation Setup}\label{Ch4_Sec5_SubSec1}
The simulation scenario in this section is configured by adopting the system model described in Section~\ref{Ch4_Sec2}. We consider an urban scenario deployment similar to the one proposed in~\cite{3GPPV2X}, where square blocks are surrounded by streets that are 250 meters long and 20 meters wide. 9 RSUs marked as red triangles are located at the crossroads and 100 user equipments (UEs) marked as blue crosses are uniformly dropped in the street in the beginning of the simulation, as illustrated in Fig.~\ref{Ch4_Fig_Simltn}. $S$ and $D$ are the upper-left RSU and the lower-right RSU depicted in the figure, respectively, and there are in total 12 loop-free routes between $S$ and $D$. Channel model (LOS probability, pathloss, blockage model, etc.) is consistent with~\cite{3GPPV2X}. The default simulation parameters are summarized in Table~\ref{Ch4_Table_Simltn}. Simulation samples are averaged over 1000 independent snapshots.
\begin{figure}[tbp]
	\centering
	\includegraphics [scale=0.5]{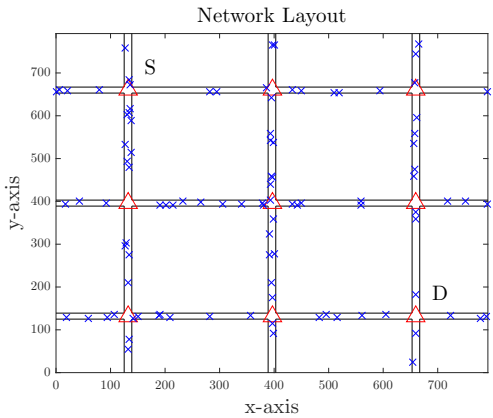}
	\caption{Illustration of simulation scenario.}\label{Ch4_Fig_Simltn}
\end{figure}
\renewcommand{\arraystretch}{0.6}
\begin{table}[tbp]
	\centering%
	\caption{Simulation Parameters}%
	\label{Ch4_Table_Simltn}%
	\begin{tabular}{|c|c|}%
		\hline%
		Parameter                                                                       & Value                                                                 \\ \hline%
		Carrier frequency                                                               & $\SI{28}{\giga\hertz}$                                                             \\ \hline%
		System bandwidth                                                                & $\SI{1}{\giga\hertz}$                                                               \\ \hline%
		Vehicle speed                                                                 & $\SI{45}{\kilo\meter/\hour}$                                          \\ \hline%
		\begin{tabular}[c]{@{}c@{}}Duration of vehicle\\ staying in each hop\end{tabular}   & $\SI{20}{\second}$\\ \hline%
		\begin{tabular}[c]{@{}c@{}}Antenna array\\ (Vertical x Horizontal)\end{tabular} & \begin{tabular}[c]{@{}c@{}}$8 \times 16$ for RSU\\ $4 \times 4$ for UE\end{tabular}      \\ \hline%
		Maximum transmission power &  \begin{tabular}[c]{@{}c@{}}$\SI{30}{\decibel m}$ for RSU\\ $\SI{23}{\decibel m}$ for UE\end{tabular}     \\ \hline%
		Decode error probability & $10^{-3}$ \\ \hline%
		Vehicle arrival rate & $[0.05,0.3]$ \\ \hline
	\end{tabular}%
\end{table}

\subsection{Performance of Global Routing Algorithm}\label{Ch4_Sec5_SubSec2}
In Fig.~\ref{Ch4_Fig_GR:a} and Fig.~\ref{Ch4_Fig_GR:b}, we plot the normalized weighted sum of E2E latency and E2E data rate, defined in~\eqref{eqnOpt}, versus the choice of global candidate discovery duration $t$ with $\alpha = 0.5$ and $\alpha = \{0,\,0.5,\,1\}$, respectively. Here, $\alpha = 0.5$ indicates the weighted sum considering both latency and data rate, and $\alpha = 0$ and $\alpha = 1$ address the E2E latency and the E2E data rate, respectively. Remember that the motivation of the normalization lies in the fact that the data ranges of latency and data rate are not directly comparable.
\begin{figure}[tbp]
	\centering
	\begin{subfigure}{0.48\columnwidth}
		\includegraphics[width=\columnwidth]{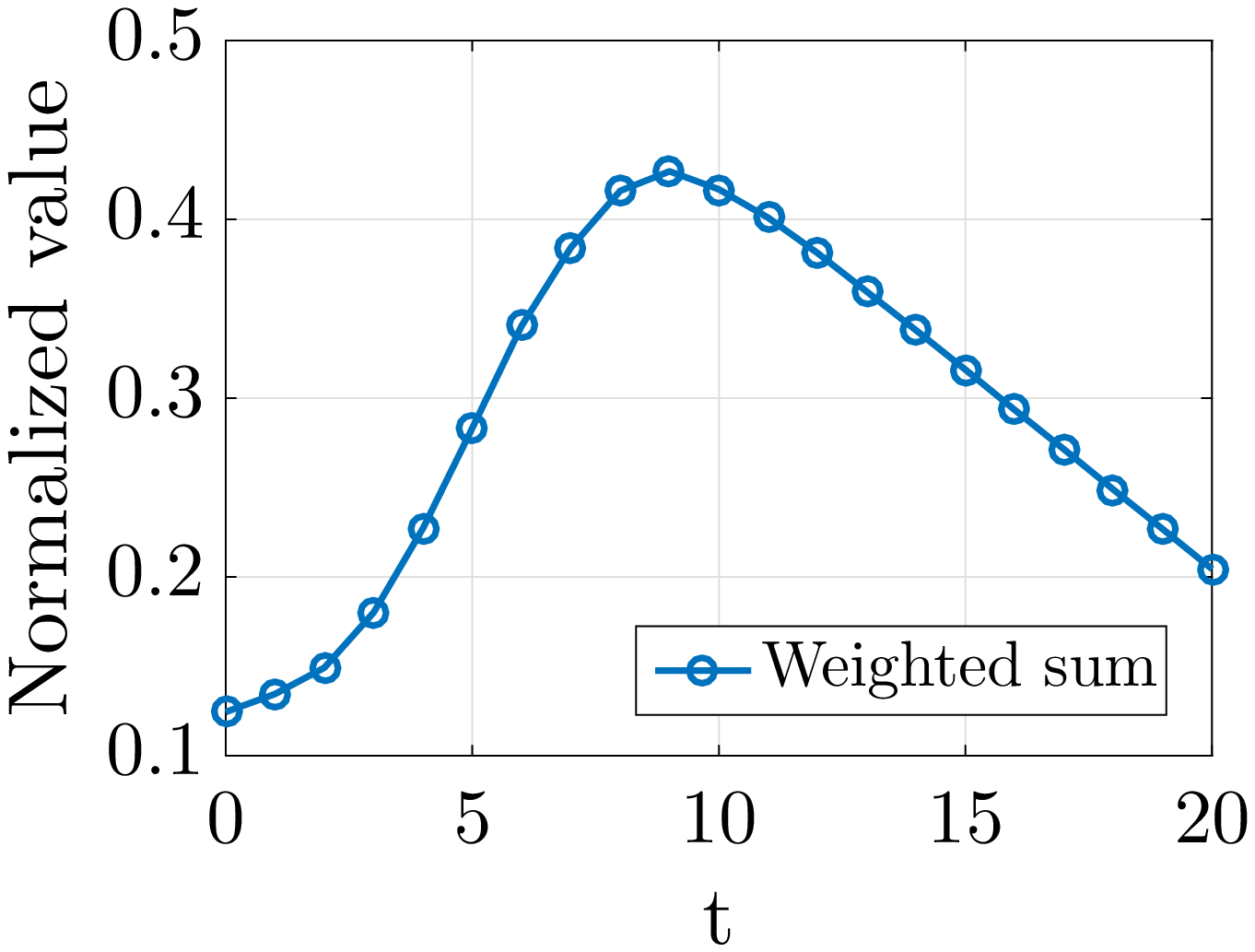}
		\caption{$\alpha = 0.5$}
		\label{Ch4_Fig_GR:a}
	\end{subfigure}
	\begin{subfigure}{0.48\columnwidth}
		\includegraphics[width=\columnwidth]{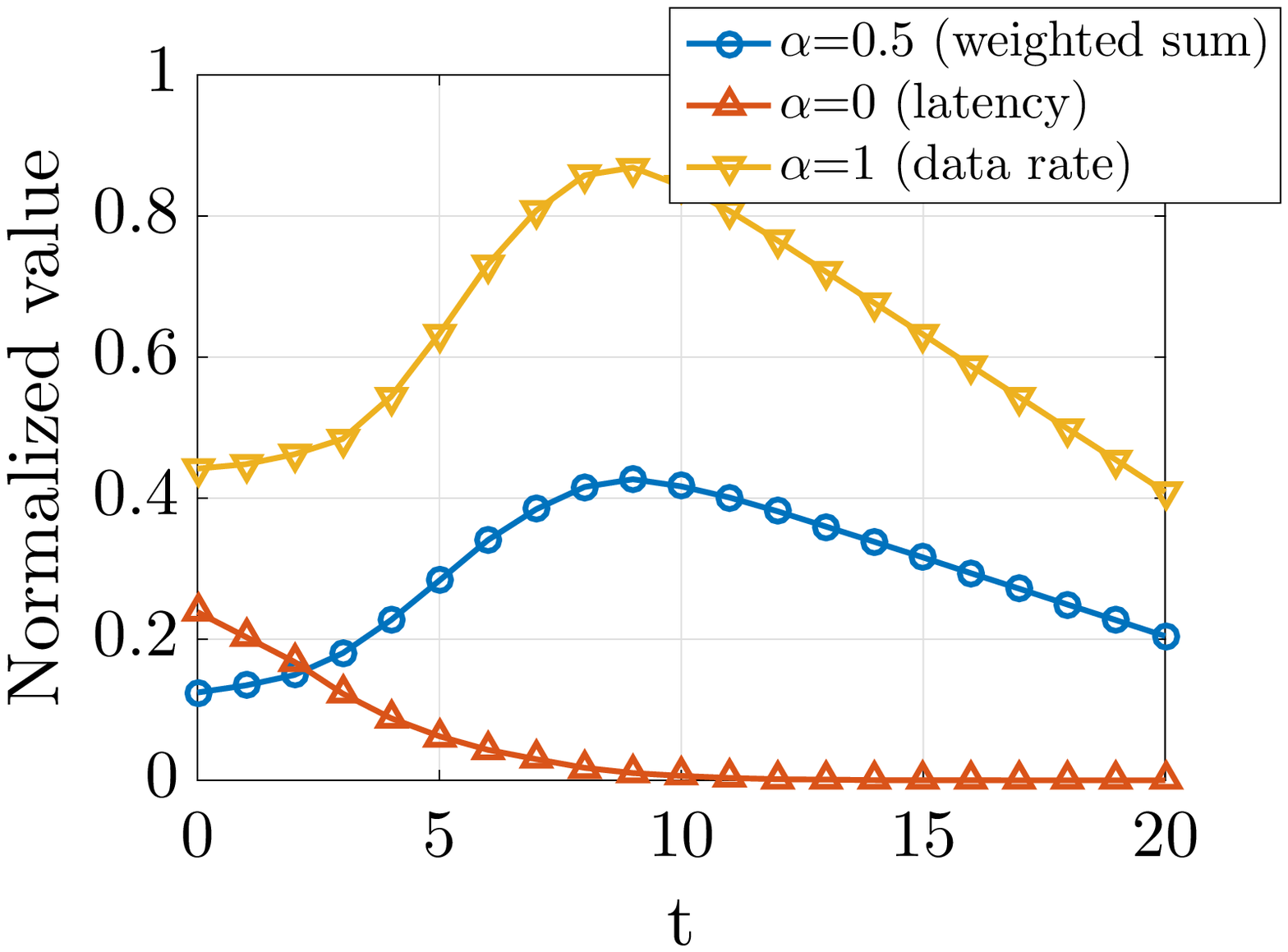}
		\caption{$\alpha = \{0,\,0.5,\,1\}$}
		\label{Ch4_Fig_GR:b}
	\end{subfigure}
	\caption{Performance of the proposed global routing algorithm is evaluated in (a) normalized weighted sum with weight $\alpha = 0.5$ and (b) normalized weighted sum with weight $\alpha = \{0,\,0.5,\,1\}$.}\label{Ch4_Fig_GR}
\end{figure}
Result in Fig.~\ref{Ch4_Fig_GR:a} shows that the solution of the global data delivery problem in Theorem~\ref{Ch4_Thm3} is accurate, where the optimal duration $t^*$ exists and maximizes the weighted sum. With an increased $t$, higher probability of successful candidate discovery can be expected where all hop latencies equal to $T$, and thus the E2E latency is reduced and optimized when $t\rightarrow T$, as shown in Fig.~\ref{Ch4_Fig_GR:b}. However, if both latency and data rate are to be considered, the optimal $t^*$ is observed in the middle range of $[0,T]$ instead of $t^* \rightarrow T$ for minimizing latency, as small value of $t$ results in failed discovery in all hops and limits $r_\text{O}t$ according to~\eqref{eqnUnfind}, while a larger $t$ degrades $r_\text{O}(T-t)$ as indicated in~\eqref{eqnUfind}.

In summary, these results indicate that the proposed global routing algorithm is efficient for solving the global optimization problem in terms of high data rate and low latency.

\subsection{Comparison of Global Routing Algorithm and Other Vehicular Routing Algorithms}\label{Ch4_Sec5_SubSec3}
In this subsection, we compare the normalized weighted sum of the proposed global routing algorithm with other classical vehicular routing algorithms. In making the comparison, we consider two well-known routing algorithms~\cite{JLiu}: the shortest path routing (SPR) that minimizes the E2E latency, and the greedy perimeter stateless routing (GPSR) which is a classical geographical-based routing algorithm with high efficiency in data delivery. 

To get some insights into the comparison, we plot the normalized weighted sum for different routing algorithms, versus the choice of candidate discovery duration $t$ in Fig.~\ref{Ch4_Fig_Algo:a} and Fig.~\ref{Ch4_Fig_Algo:b}, with $\alpha = \{0.5,\,1\}$ and $\alpha = 0$ (latency), respectively. 
\begin{figure}[tbp]
	\centering
	\begin{subfigure}{0.48\columnwidth}
		\includegraphics[width=\columnwidth]{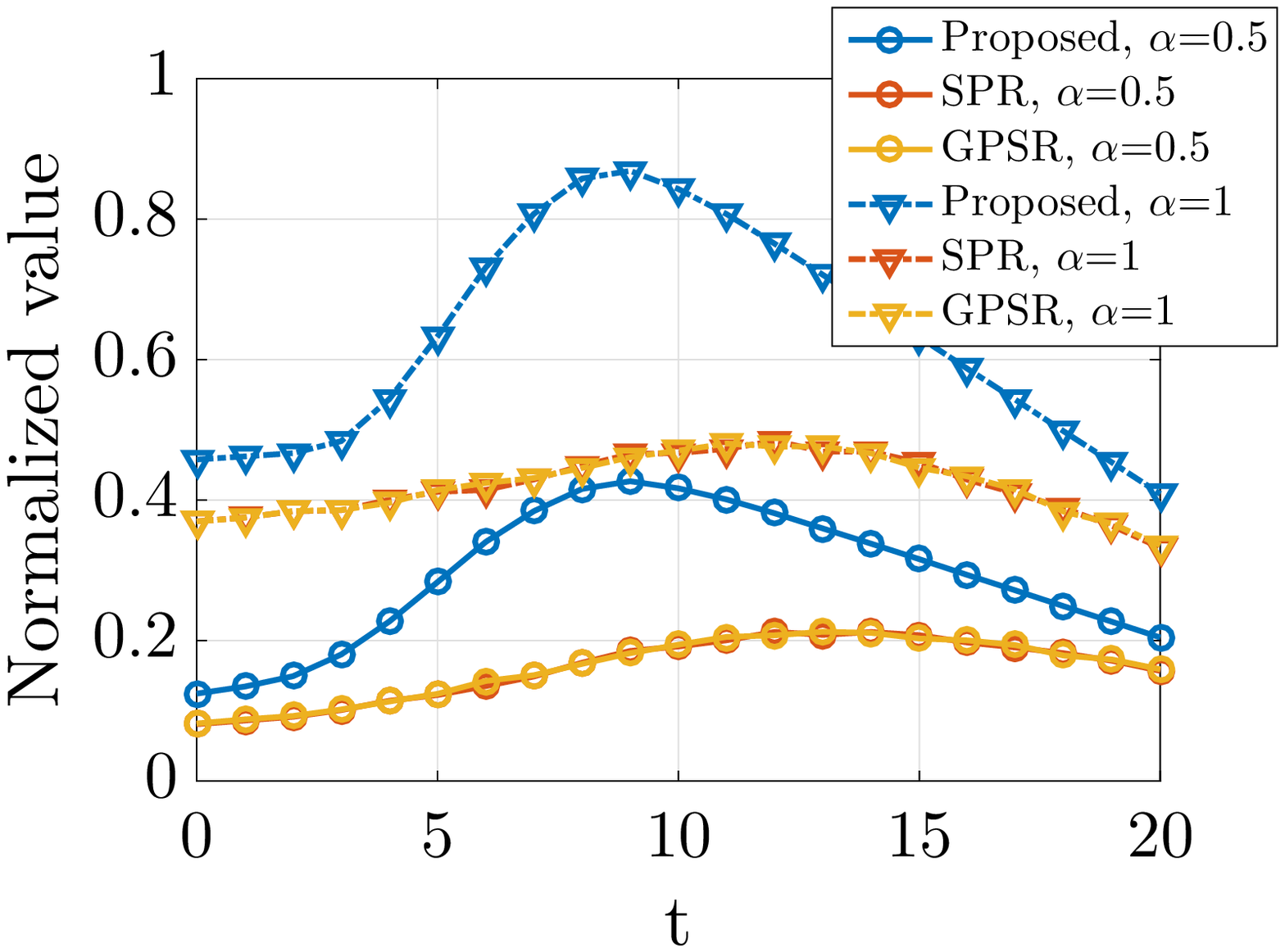}
		\caption{$\alpha = \{0.5,\,1\}$}
		\label{Ch4_Fig_Algo:a}
	\end{subfigure}
	\begin{subfigure}{0.48\columnwidth}
		\includegraphics[width=\columnwidth]{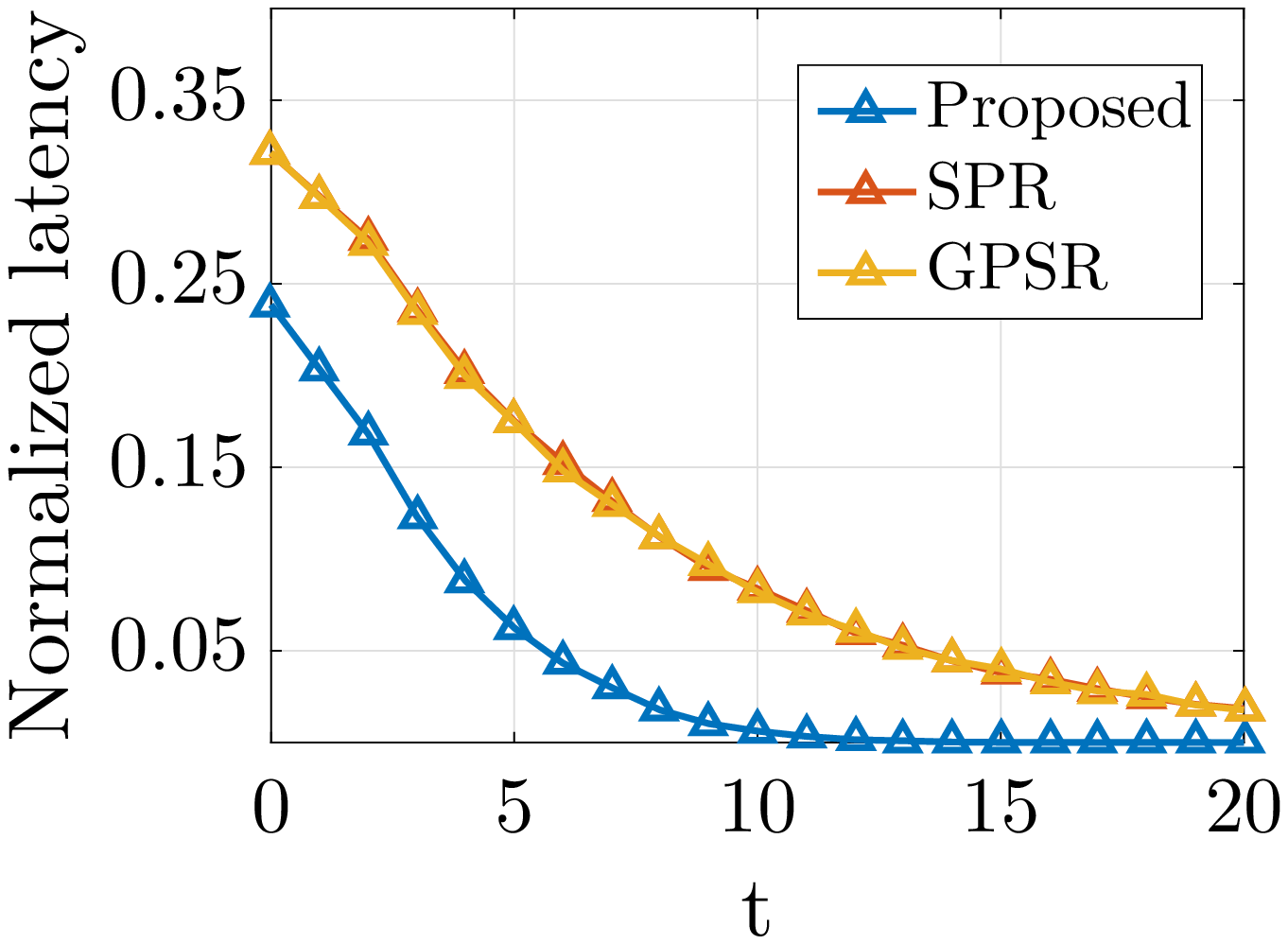}
		\caption{$\alpha = 0$}
		\label{Ch4_Fig_Algo:b}
	\end{subfigure}
	\caption{Performance of the proposed global routing algorithm and other routing algorithms are compared in (a) normalized weighted sum with weight $\alpha = \{0.5,\,1\}$ and (b) normalized E2E latency.}\label{Ch4_Fig_Algo}
\end{figure}
On the one hand, as depicted in Fig.~\ref{Ch4_Fig_Algo:a}, the proposed global algorithm outperforms the others in terms of the weighted sum with weight $\alpha = 0.5$ and $\alpha = 1$ (data rate), as the proposed algorithm selects the route that maximizes the weighted sum, which shows better performance than selecting the shortest route by SPR and GPSR. On the other hand, the proposed algorithm also achieves lower E2E latency compared to SPR and GPSR that are supposed to be able to minimize the latency, as shown in Fig.~\ref{Ch4_Fig_Algo:b}. The reason behind lies in the fact that the geographically shortest route selected by SPR and GPSR is not necessarily the route that minimizes the latency.

\subsection{Impact of Broadcast Schemes and Vehicle Arrival Rates on Data Delivery Performance}\label{Ch4_Sec5_SubSec4}
Fig.~\ref{Ch4_Fig_Brd:a} and Fig.~\ref{Ch4_Fig_Brd:b} plot the maximal normalized weighted sum calculated from~\eqref{eqnOpt} for different broadcast schemes versus the number of simultaneous beams $M$ for beamforming, with $\alpha = 0.5$ and $\alpha = 1$ (data rate), respectively. Here, TD, FD, CD, and SD represent the broadcast schemes addressed in~\cite{YLiTWC} and multiplexed in time, time-frequency, time-code, and time-space, respectively. TD refers to the single beam exhaustive scan where $M=1$, hence we plot a single circle instead of a curve to represent the maximal weighted sum in Fig.~\ref{Ch4_Fig_Brd:a} and Fig.~\ref{Ch4_Fig_Brd:b}. For other schemes, multiple beams are formed simultaneously. Note that the E2E latency with $\alpha = 0$ is not compared as the minimal normalized latency is always 0.
\begin{figure}[tbp]
	\centering
	\begin{subfigure}{0.48\columnwidth}
		\includegraphics[width=\columnwidth]{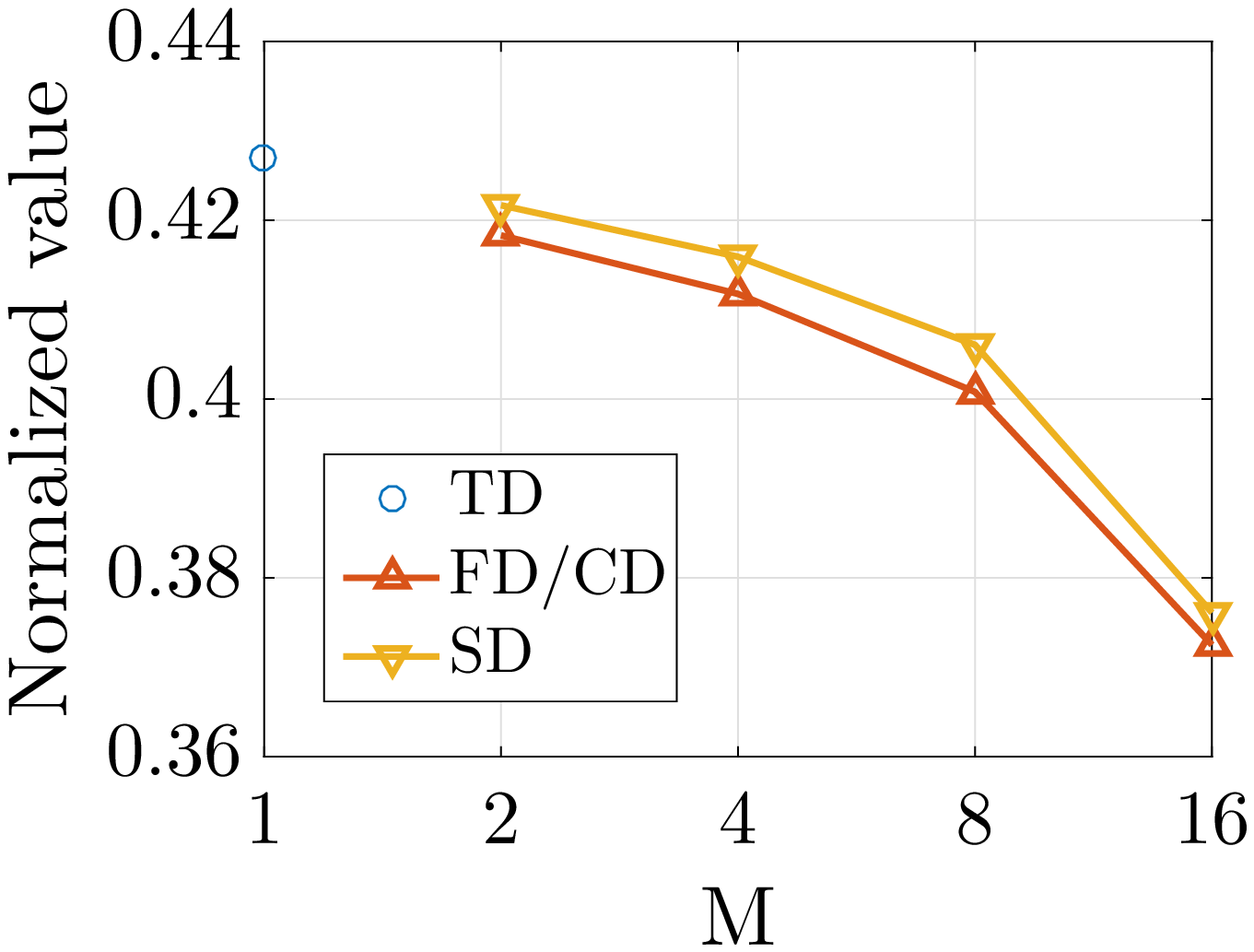}
		\caption{$\alpha = 0.5$}
		\label{Ch4_Fig_Brd:a}
	\end{subfigure}
	\begin{subfigure}{0.48\columnwidth}
		\includegraphics[width=\columnwidth]{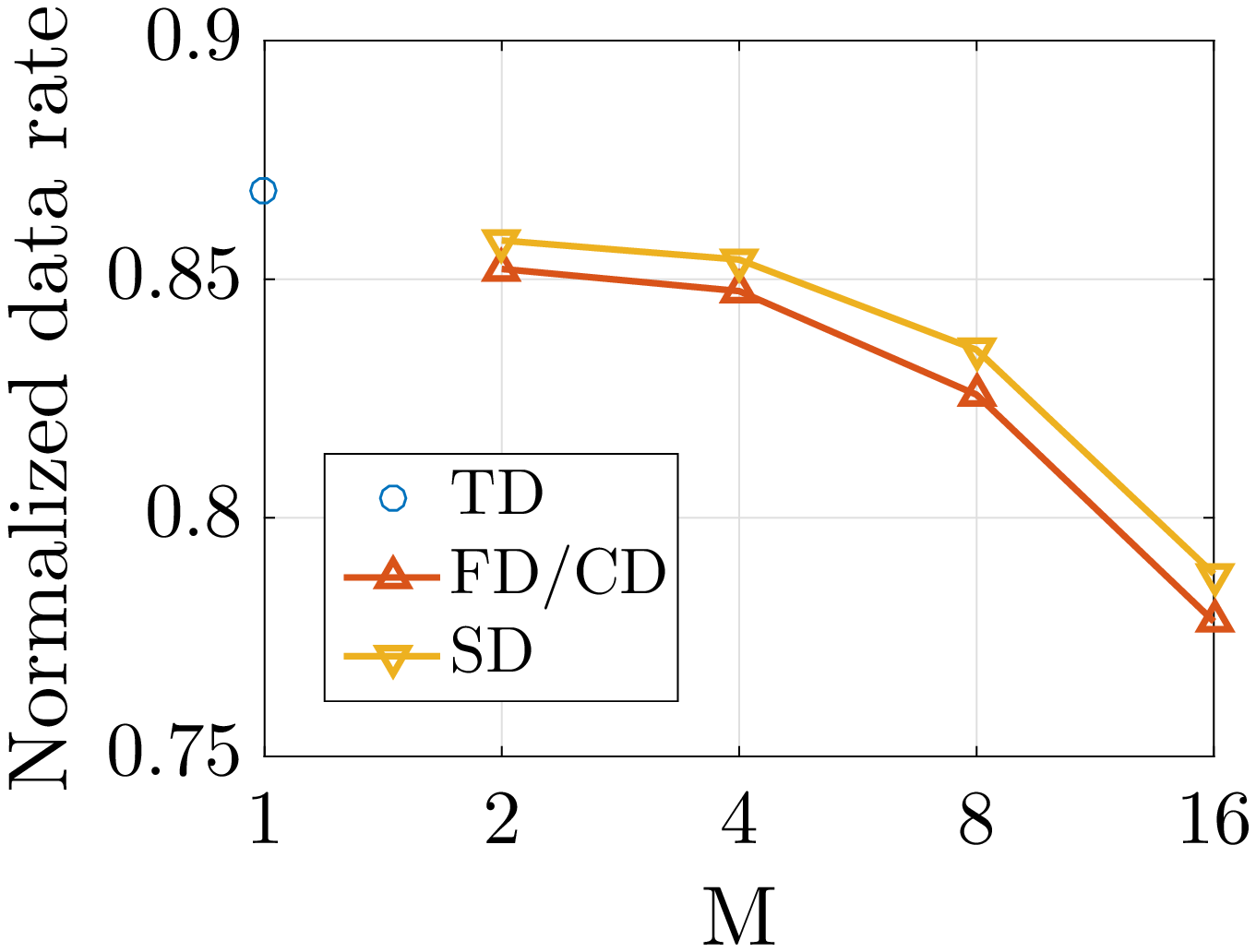}
		\caption{$\alpha = 1$}
		\label{Ch4_Fig_Brd:b}
	\end{subfigure}
	\caption{Performance of different broadcast schemes for candidate discovery are compared in (a) maximal normalized weighted sum with weight $\alpha = 0.5$ and (b) maximal normalized E2E data rate.}\label{Ch4_Fig_Brd}
\end{figure}

According to the conclusion in~\cite{YLiTWC}, TD achieves the lowest average discovery latency which means that if the courier applies single beam exhaustive scan for candidate discovery, the time of one discovery trial $\Delta t$ is minimized compared with other broadcast schemes. Moreover, FD and SD perform exactly same and worse than TD in terms of discovery latency, and the curve of SD locates in-between of TD and FD/CD.

In this way, TD decreases the number of maximal discovery trials $m$ and the actual candidate discovery time $\tau_{h,h+1}^{(d, \text{Veh})}$, and therefore leads to higher probability of successful candidate discovery, lower E2E latency, and higher hop-wise data rate of successful candidate discovery that are addressed in~\eqref{eqnPfind},~\eqref{eqnDhop}, and~\eqref{eqnUfind}, respectively, which eventually achieves the highest weighted sums for both $\alpha = 0.5$ and $\alpha = 1$, as illustrated in Fig.~\ref{Ch4_Fig_Brd:a} and Fig.~\ref{Ch4_Fig_Brd:b}, respectively. By contrast, the weighted sums degrade when SD and FD/CD, which call for discovery latency penalty, are utilized for candidate discovery with more simultaneous beams. Therefore, the results in Fig.~\ref{Ch4_Fig_Brd:a} and Fig.~\ref{Ch4_Fig_Brd:b} recommend TD as the optimal beamforming and broadcast scheme for candidate discovery to explicitly achieve the best data delivery performance.
\begin{figure}[tbp]
	\centering
	\begin{subfigure}{0.48\columnwidth}
		\includegraphics[width=\columnwidth]{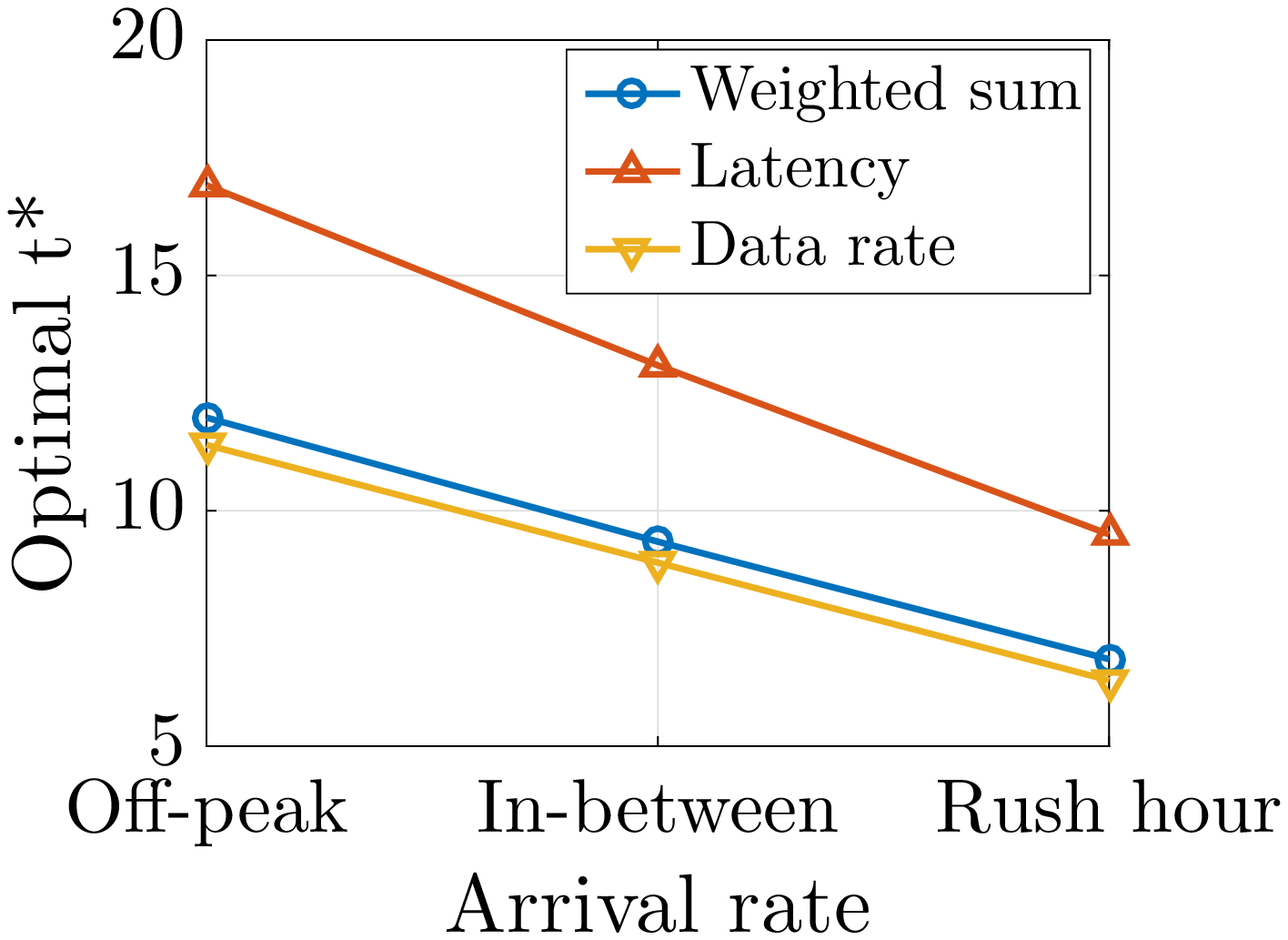}
		\caption{Optimal $t^*$}
		\label{Ch4_Fig_ArrBH:a}
	\end{subfigure}
	~ 
	\begin{subfigure}{0.48\columnwidth}
		\includegraphics[width=\columnwidth]{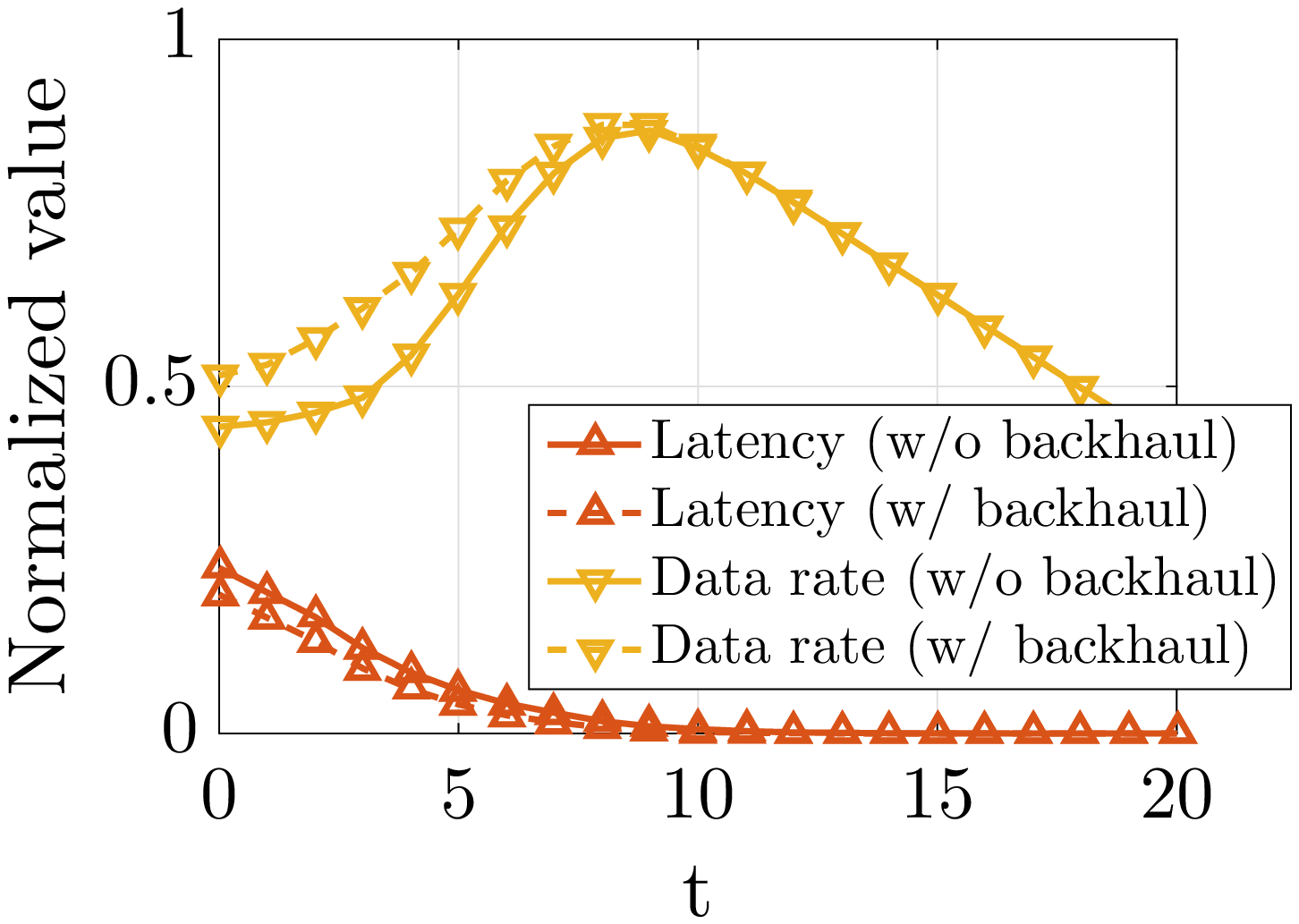}
		\caption{Impact of backhaul}
		\label{Ch4_Fig_ArrBH:b}
	\end{subfigure}
	\caption{Optimal $t^*$ to achieve the maximal weighted sum for different vehicle arrival rates are compared in (a). Performance of data delivery with and without backhaul are compared in (b).}\label{Ch4_Fig_ArrBH}
\end{figure}

Fig.~\ref{Ch4_Fig_ArrBH:a} plots the optimal candidate discovery duration $t^*$ that optimizes the weighted sum with weight $\alpha = 0.5$ for different vehicle arrival rates. Here, the intervals of arrival rates for off-peak, in-between, and rush hour are defined as $[0.05,0.15]$, $[0.1,0.2]$, and $[0.2,0.3]$ vehicles/second, respectively. According to the results, we conclude that the optimal $t^*$ is almost linearly decreased. The reason behind this falls in the fact that a higher arrival rate improves the probability of successful candidate discovery in~\eqref{eqnPfind}, and correspondingly a relative small $t$ is able to achieve the best data delivery performance. 

\subsection{Performance of Wireless Backhauling for Data Delivery}\label{Ch4_Sec5_SubSec5}
Wireless backhauling, as an efficient alternative to expensive fiber connectivity among RSUs, provides coverage extension and capacity expansion with low latency, as addressed in Section~\ref{Ch4_Sec1}. Therefore, instead of discovering a candidate to forward data received from courier, RSU equipped with wireless backhaul is able to transmit the data directly to RSU of the next hop in case they are interconnected via wireless backhaul. In Fig.~\ref{Ch4_Fig_ArrBH:b}, the normalized E2E latency ($\alpha = 0$) and the normalized E2E data rate ($\alpha = 1$) of global data delivery with (w/) backhaul and without (w/o) backhaul are compared, versus the choice of global candidate discovery duration $t$.

On the one hand, the results in Fig.~\ref{Ch4_Fig_ArrBH:b} suggest that incorporating backhaul data forwarding reduces the latency when $t$ is small, i.e., the support of backhaul alleviates the burden of RSU in candidate discovery when courier fails. On the other hand, increased data rates are also observed when $t \rightarrow 0$, because higher data rate is provided by the backhaul transmission compared to V2I communications. Note that for large values of $t$, candidate discovery tends to be successful in almost all hops, therefore the data is carried and forwarded purely by vehicles without store-and-forward from RSU, and data delivery w/ and w/o backhaul perform exactly the same.

\subsection{Comparison of Global and Distributed Routing Algorithm}\label{Ch4_Sec5_SubSec6}
In the end, the performance of global and distributed routing algorithms are compared in Table~\ref{Ch4_Table_CGD}, Fig.~\ref{Ch4_Fig_hist:a}, and Fig.~\ref{Ch4_Fig_hist:b}. It is stated in Table~\ref{Ch4_Table_CGD} that around $4.7\%$ and $2.9\%$ gain of the maximal normalized weighted sum with weight $\alpha = 0.5$ and $\alpha = 1$ can be achieved by the distributed algorithm compared to the global one, respectively. The normalized E2E latency are minimized by both algorithms with equal value of 0 and thus are not recorded in the table.
\begin{figure}[tbp]
	\centering
	\begin{subfigure}{0.48\columnwidth}
		\includegraphics[width=\columnwidth]{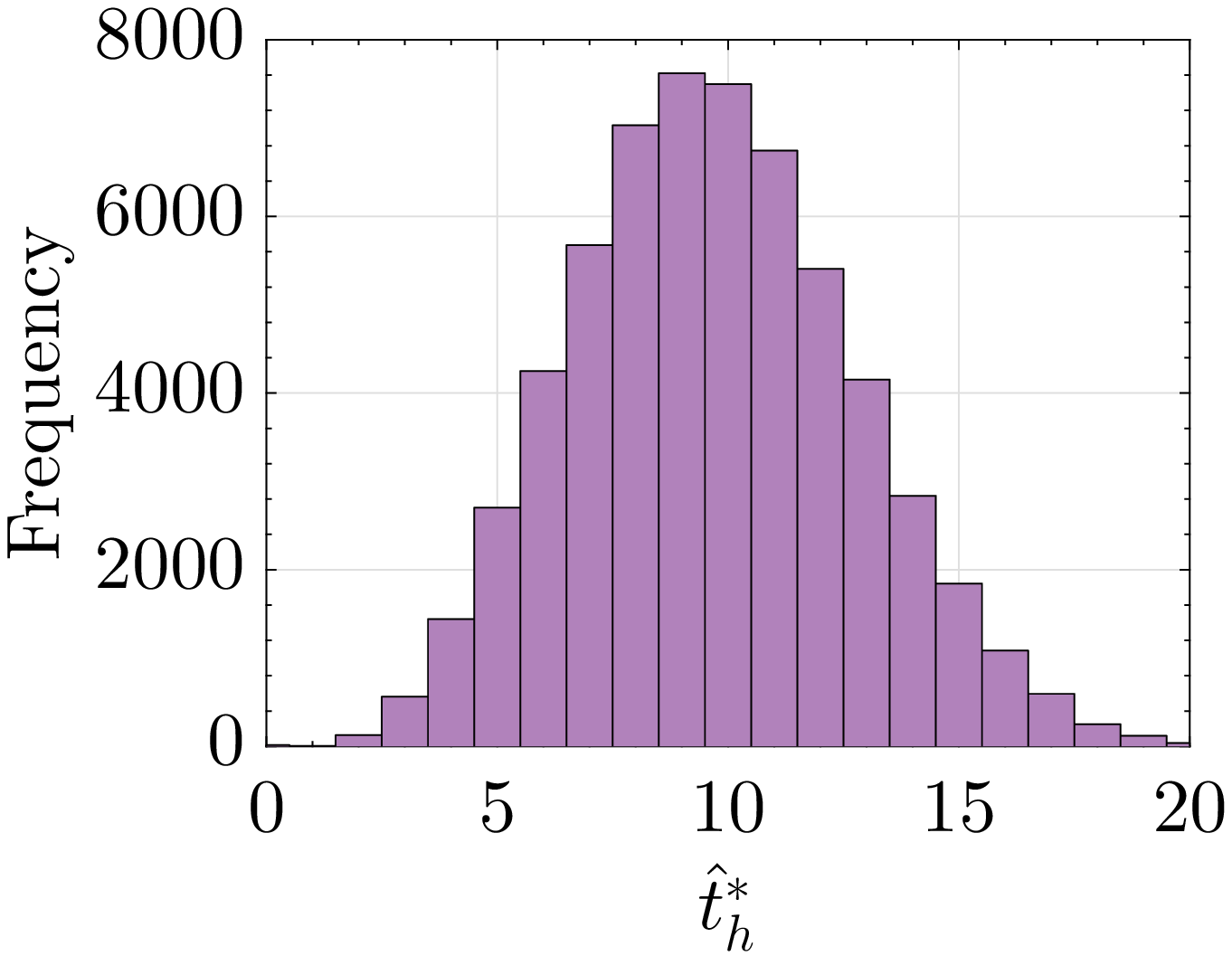}
		\caption{Latency}
		\label{Ch4_Fig_hist:a}
	\end{subfigure}
	\begin{subfigure}{0.48\columnwidth}
		\includegraphics[width=\columnwidth]{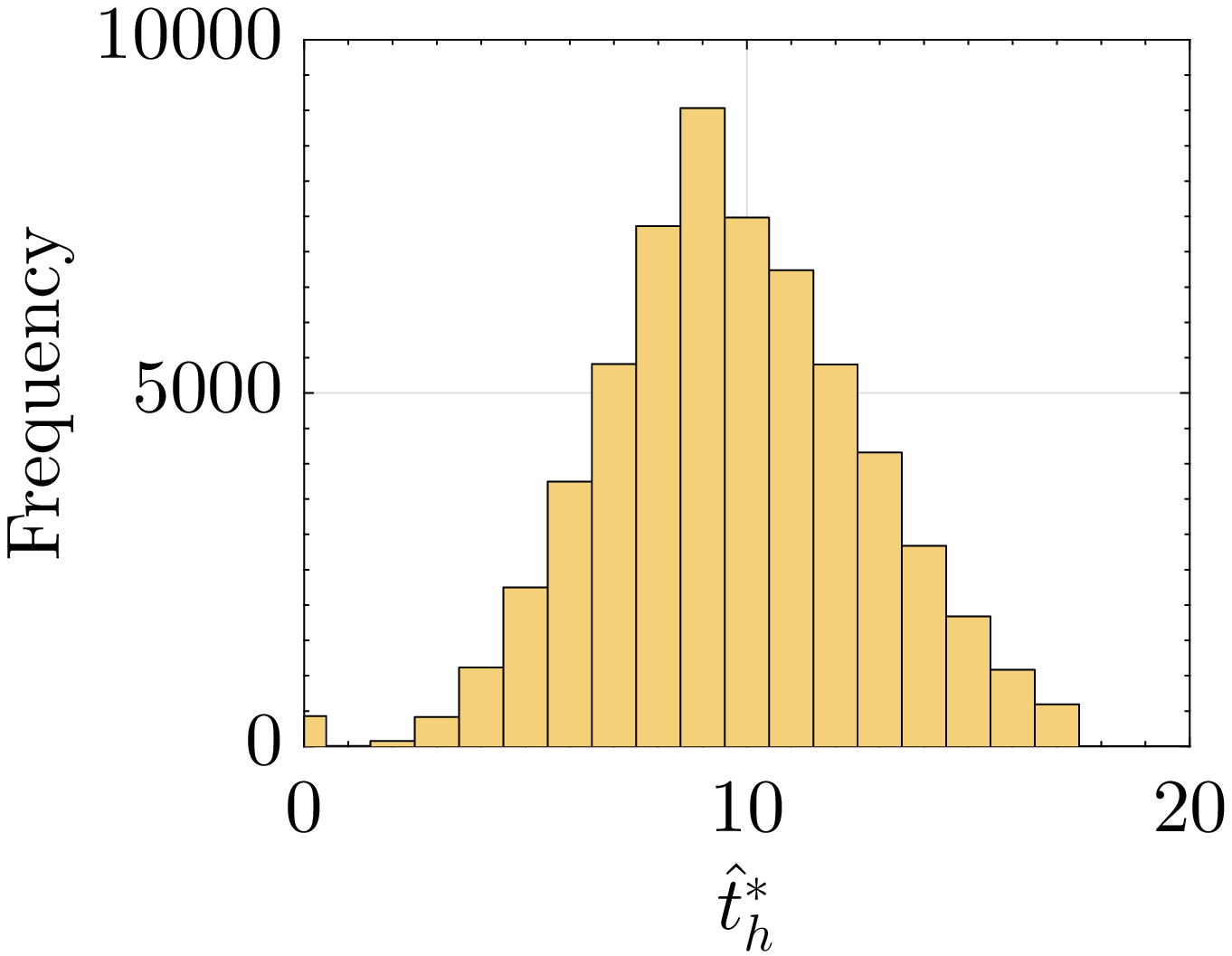}
		\caption{Data rate}
		\label{Ch4_Fig_hist:b}
	\end{subfigure}
	\caption{Illustration of histogram of optimal hop-wise candidate discovery duration $\hat{t}^*_h$ for (a) minimizing hop-wise latency and (b) maximizing hop-wise data rate.}\label{Ch4_Fig_hist}
\end{figure}
\renewcommand{\arraystretch}{0.6}
\begin{table}[tbp]
	\centering
	\caption{Comparison of Global and Distributed Routing Algorithm}
	\label{Ch4_Table_CGD}
	\begin{tabular}{|l|l|l|l|}
		\hline
		& Global & Distributed & Gain  \\ \hline
		$\alpha = 0.5$ & 0.4269 & 0.4469      & 4.7\% \\ \hline
		$\alpha = 1$   & 0.8688 & 0.8938      & 2.9\% \\ \hline
	\end{tabular}
\end{table}

Fig.~\ref{Ch4_Fig_hist:a} and Fig.~\ref{Ch4_Fig_hist:b} plot the histogram of the optimal hop-wise candidate discovery duration $\hat{t}^*_h$ for the distributed routing algorithm obtained by~\eqref{eqnSold1} across different snapshots for minimal E2E latency and maximal E2E latency, respectively. We notice that both histogram figures yield distributions of optimal $\hat{t}^*_h$ similar to the curve in Fig.~\ref{Ch4_Fig_GR:b}. Specifically, the most counts of $\hat{t}^*_h$ for both minimal latency and maximal data rate emerge at $\hat{t}^*_h=9$, which is in consistence with the observation from Fig.~\ref{Ch4_Fig_GR:b} where $t^*=9$ is the optimal value for minimizing latency and maximizing data rate with the global routing algorithm.

\section{Summary}\label{Ch4_Sec6}
In this paper, we investigated the data delivery problem in V2X networks. We proposed an analytical framework to derive the mathematical expressions of expected delivery latency and data rate. Based on this, optimization problems for both global and distributed data delivery have been formulated to maximize the weighted sum of E2E latency and data rate and the weighted sum of hop-wise latency and data rare, respectively. Leveraging the reformation of closed-form expressions of expected latency and data rate, the convexity of the proposed optimization problems are verified and optimal solutions that maximize the weighted sum of both E2E and hop-wise latency and data rate are proposed. Afterwards, both global and distributed multihop routing algorithms are developed for solving the global and distributed optimization problems by obtaining optimal global and hop-wise candidate discovery durations, respectively. Finally, the extensive system-level simulations are conducted to evaluate the performance of the proposed algorithms by considering different vehicular routing algorithms, various beamforming configurations and vehicular arrival rates, and the support of wireless backhauling. It is demonstrated that the weighted sum is maximized by obtaining the optimal candidate discovery duration which varies with desired optimization objectives (latency and/or data rate). This implies that a large candidate discovery duration is recommended for reducing the latency. If both latency and data rate are to be considered, a relatively smaller value of candidate discovery duration provides the flexibility to achieve a trade-off between latency and data rate. 
Furthermore, the proposed routing algorithm provides considerable improvement over classical vehicular routing algorithms in the sense of minimizing latency while maximizing data rate. Finally, the selection of different broadcast schemes and vehicle arrival rates, and the availability of wireless backhauling, are demonstrated to influence latency and data rate, and the distributed multihop routing algorithm is shown to be able to provide lower latency and higher data rate compared to the global algorithm.

Future work can leverage the proposed analytical framework to investigate data delivery performance under other metrics, e.\,g., delivery ratio, or to extend the framework by incorporating more accurate radio model with relaxed assumptions. It would also be interesting to include various system models such as high-way, and/or different vehicular traffic types.

\appendices

\section{Proof of Convexity of Data Delivery Optimization Problem}\label{Ch4_Sec7_SubSec6}

To prove the convexity of the global data delivery optimization problem, we incorporate the following conclusions addressed in~\cite{Boyd}:
\begin{enumerate}[nolistsep]
	\item Exponential function is convex: $e^{ax}, \forall a \in \mathbb{R}$. \label{list:1}
	\item Power function is convex: $x^a$, $\forall a \geq 1$. \label{list:2}
	\item Weighted sum of convex/concave functions is convex/concave: $\sum_{i=1}^{n}\omega_i f_i$. \label{list:3}
	\item Point-wise maximum/minimum of convex/concave functions is convex/concave: $\max_{\forall i} f_i$. \label{list:4}
\end{enumerate}
Based on these, we transform the expected E2E latency derived in~\eqref{eqnDrouteNew} as follows:
\begin{align*}
\bar{L} &= kT + \sum_{h=1}^{k} (1-\alpha_h)\phi_h\big(\beta_h(t)+\theta_h(t)-\beta_h(t)\theta_h(t)\big) \\
&= A + \sum_{h=1}^{k} B\big(\beta_h(t)+\theta_h(t)-\beta_h(t)\theta_h(t)\big). \numberthis
\label{eqnDrouteNewConvex}
\end{align*}
Here, $A=kT$ and $B= (1-\alpha_h)\phi_h$ can be treated as constant without the term $t$. $\beta_h(t)$ and $\theta_h(t)$ are functions of $t$ and convex according to~\ref{list:1}). Then, the expected E2E latency $\bar{L}$, calculated as the weighted sum of convex function $\beta_h(t)$, $\theta_h(t)$, and product $\beta_h(t)\theta_h(t)$, is convex as stated in~\ref{list:3}).

For the expected E2E data rate, similar demonstration of convexity can be deducted. Firstly, note that $\bar{C}$ in~\eqref{eqnUroute} can be reformed as
\begin{align*}
\bar{C}
&= \min_{\forall h} \Big( \zeta_h + (1-\alpha_h)\big(1-z(t)\big)  \big(\iota_h + \nu_h(t) \big)\\
& \qquad \qquad + (1-\alpha_h)z(t) \big( \kappa_h + \chi_h(t) \big) \Big). \numberthis
\label{eqnUrouteConvex1}
\end{align*}
By taking the second-order derivative, it can be shown that both $\big(1-z(t)\big) \big(\iota_h + \nu_h(t) \big)$ and $z(t) \big( \kappa_h + \chi_h(t) \big)$ are concave. Hence, the weighted sum of these two concave functions and the constant $\zeta_h$ is also concave according to~\ref{list:3}), and the min operator $\min_{\forall h} (.)$ preserves the concavity of the weighted sum according to~\ref{list:4}). Detailed derivations are omitted here due to limited space.

Finally, $\alpha \bar{C} - (1-\alpha) \bar{L}$, represented by the weighted sum of concave function $-\bar{L}$ and concave function $\bar{C}$, is a concave function in line with~\ref{list:3}). In this way, the convexity of the optimization problem~\ref{Ch4_prb1} is verified, and the weighted sum of the E2E latency and the E2E data rate is maximized by determining the optimal $t^*$. 

\section{Proof of Theorem~\ref{Ch4_Thm3}}\label{Ch4_Sec7_SubSec7}

The optimization problem~\ref{Ch4_prb1} can be refined as
\begin{equation}
\max_{\substack{t \in \mathbb{R}}} \alpha \bar{C} - (1-\alpha) \bar{L} \quad \text{s.t.} \quad t\geq 0, \; t \leq T.
\label{eqnSolc2}
\end{equation}
Observe that all inequality constraints functions are affine, and there exists a value of $t$ such that $t=\frac{T}{2} \geq 0$ and $t=\frac{T}{2} \leq T$, which implies the Karush-Kuhn-Tucker (KKT) conditions are necessary. Moreover, the objective function is demonstrated as convex in Appendix~\ref{Ch4_Sec7_SubSec6}. Therefore, the KKT conditions are also sufficient, in which we can use standard KKT form to solve the problem. They are concluded as follows:
\begin{itemize}[noitemsep]
	{\item PF
		\begin{equation}
		t\geq 0, \; t \leq T.
		\label{eqnSolc3}
		\end{equation}}
	{\item DF
		\begin{equation}
		-\frac{\mathrm{d}(\alpha \bar{C} - (1-\alpha) \bar{L})}{\mathrm{d}t} - \psi + \omega = 0,\; \psi \geq 0,\; \omega \geq 0.
		\label{eqnSolc4}
		\end{equation}}
	{\item CS
		\begin{equation}
		\psi t = 0,\; \omega t = 0.
		\label{eqnSolc5}
		\end{equation}}
\end{itemize}
As the value of $t$ is restricted in the set $[0,T]$, we can always take any $t$ without violate PF conditions in case $t>0$ and $t<T$. Therefore, both CS conditions are always satisfied and there is no further restriction on the Lagrangian multiplier $\psi$ and $\omega$ besides $\psi \geq 0$ and $\omega \geq 0$ from DF conditions. Actually, CS conditions also indicate that $\psi=0$ and $\omega=0$ for $t \in (0,T]$, where in order to satisfy DF conditions, $\frac{\mathrm{d}(\alpha \bar{C} - (1-\alpha) \bar{L})}{\mathrm{d}t}|_{t=t_s}$ must be zero, which leads to $t_s \in [0,T]$. 
In conclusion, the optimal solution of the problem in~\eqref{eqnOpt} is given by
\begin{equation}
t^* = \argmax_{t=\{0,t_s,T\}} \alpha \bar{C} - (1-\alpha) \bar{L}, t_s \in [0,T].
\label{eqnSolc6}
\end{equation}

\bibliographystyle{IEEEtran}
\bibliography{mm_wave}

\end{document}